\documentclass[11pt,a4paper]{article}

\usepackage[utf8]{inputenc}
\usepackage[T1]{fontenc}
\usepackage{lmodern}
\usepackage[margin=1in]{geometry}
\usepackage{amsmath,amssymb,amsfonts,amsthm,mathtools}
\usepackage{bm}
\usepackage{graphicx}
\usepackage{booktabs}
\usepackage{array}
\usepackage{multirow}
\usepackage{xcolor}
\usepackage{enumitem}
\usepackage{float}
\usepackage{morefloats}
\usepackage{caption}
\usepackage{placeins}
\usepackage[numbers,sort&compress]{natbib}
\usepackage[colorlinks=true,linkcolor=blue!70!black,citecolor=blue!70!black,urlcolor=blue!70!black]{hyperref}

\theoremstyle{plain}
\newtheorem{theorem}{Theorem}

\theoremstyle{definition}

\DeclareMathOperator{\vecop}{vec}
\DeclareMathOperator{\Null}{Null}
\DeclareMathOperator{\SNF}{SNF}
\DeclareMathOperator{\Aut}{Aut}
\DeclareMathOperator{\Spec}{Spec}
\DeclareMathOperator{\diag}{diag}
\DeclareMathOperator{\rank}{rank}

\title{\textbf{Subgraph Isomorphism, Integer Linear Programming, and Lattices:\\A Linear Algebraic Framework}}
\author{\textbf{Bryant Gipson}\\Google}
\date{}

\begin{document}

\maketitle

\begin{abstract}
\noindent
The subgraph isomorphism problem (SGI) asks a deceptively simple question: is a graph $B$ found inside a parent graph $A$? Famously NP-complete and the subject of volumes of study, attempts to practically solve subgraph problems typically require specialized software or the focused application of problem-specific theoretical knowledge. Here we present a straightforward linear framework that neatly joins SGI to lattice theory and constrained satisfiability and, importantly, yields an easily implemented algorithm that naturally leverages consumer-grade hardware and off-the-shelf libraries to compete with state-of-the-art specialized solvers. We additionally show practical realizations in approachable settings, along with some benchmarks that reduce 6+ hour computations to seconds. Aimed at readability, we develop a series of computable algebraic invariants and demonstrate their application to familiar problems such as the Knight's Tour, with a focus on providing an accessible entry point for researchers in graph theory, abstract algebra, and students in general.
\end{abstract}

\section{Introduction}
\label{sec:intro}

The subgraph isomorphism (SGI) problem asks an apparently simple question: when is one graph contained inside another? Framed another way, given a pattern graph $G_B$ on $M$ vertices and a target graph $G_A$ on $N \ge M$ vertices, how hard is it to find a copy of $G_B$ inside $G_A$---or prove that no such copy exists? Easy to state, subgraph isomorphism is famously NP-complete \cite{garey1979computers}, generalizing both the graph isomorphism problem \cite{mckay2014practical,babai2016graph,grohe2020graph} and the Hamiltonian cycle problem, with direct relevance to chemical informatics, circuit layout, computer vision, and network analysis.

While both general and constrained versions of this problem have been extensively studied using deep insights from combinatorics, backtracking search (\texttt{Ullmann} \cite{ullmann1976algorithm}, \texttt{VF2} \cite{cordella2004subgraph}, \texttt{RI} \cite{bonnici2013subgraph}, \texttt{VF3} \cite{carletti2018challenging}, \texttt{FASTiso} \cite{fastiso2023}), constraint programming and cutting planes (\texttt{LAD} \cite{solnon2010alldifferent}, \texttt{Glasgow} \cite{mccreesh2015shape,mccreesh2020glasgow,gocht2020subgraph}), and database indexing (\texttt{TurboISO} \cite{han2013turboiso}, \texttt{DAF} \cite{han2019daf}), these algorithmic paradigms are typically developed separately from the spectral and Smith normal form ($\SNF$) invariants of algebraic graph theory \cite{godsil2001algebraic,wang2006sufficient,stanley2016smith,abiad2022codeterminantal,qiu2023smith}---which, conversely, have historically been restricted to square ($M = N$) graph isomorphism and cospectrality.

This paper bridges these two worlds through the lens of \textbf{linear algebra, integer linear programming (ILP), and lattice theory}. By treating the rectangular matrix commutator $P A - B P$ as a global linear operator rather than a collection of decoupled Boolean clauses, we:
\begin{enumerate}[leftmargin=1.8em]
    \item unify square graph isomorphism ($M = N$) and rectangular subgraph isomorphism ($M \le N$, both induced and non-induced) in a single $MN$-variable Kronecker operator framework (Section~\ref{sec:linear_algebra}) that naturally exploits modern multi-core CP-SAT and ILP solvers (Sections~\ref{sec:ilp_sgi}--\ref{sec:benchmarks_applications});
    \item construct an explicit unimodular cycle basis $L$ for the doubly-stochastic lattice that connects the commutator kernel directly to generalized cospectrality and affine lattice norm minimization (Sections~\ref{sec:spectral}--\ref{sec:snf_sieve}); and
    \item introduce a projection of the rectangular commutator onto the $(M-1)(N-1)$-dimensional \textbf{zero-sum difference lattice $\mathcal{L}^{(0)}_{M,N}$} that eliminates unmatched vertex asymmetries and boundary edges, yielding a new \textbf{Subgraph $\SNF$ Interlacing Theorem} (Theorem~\ref{thm:snf_interlacing}) and polynomial-time refutation certificates for Kuratowski obstructions, longest cycles, and Knight's Tours (Sections~\ref{sec:snf_sgi}--\ref{sec:knights_tour}).
\end{enumerate}

Aimed at approachability, the goal is to show how the classical linear commutator of graph isomorphism extends directly to rectangular subgraph isomorphism, offering a clean set of algebraic and computational in-roads to the problem with as simple a presentation as possible.

\section{Subgraph Isomorphism and Linear Algebra}
\label{sec:linear_algebra}

\subsection{From Square Permutations to the Kronecker Commutator ($M = N$)}
\label{sec:kronecker_gi}

Consider two graphs, $G_A$ and $G_B$, with corresponding adjacency matrices $A \in \{0,1\}^{N \times N}$ and $B \in \{0,1\}^{M \times M}$, respectively. In the special case where both graphs have the same number of vertices ($M = N$), Graph isomorphism between $G_A$ and $G_B$ ($G_A \simeq G_B$) holds if and only if there exists a square permutation matrix $P \in \{0,1\}^{N \times N}$ such that:
\begin{equation*}
P A P^T = B
\end{equation*}
Because $P$ is orthogonal ($P^T = P^{-1}$), multiplying on the right by $P$ yields the classical linear commutator equation:
\begin{equation}
\label{eq:commutator}
P A - B P = 0
\end{equation}
Though previous work has often framed Equation~\eqref{eq:commutator} as a doubly-stochastic continuous relaxation \cite{tinhofer1986graph,ramana1994fractional,arvind2017color} or via $\ell_1$ penalty linearizations for weighted graph matching and edit distance \cite{almohamad1993linear,justice2006binary}, one can utilize the vectorization operator ($\vecop$) and the Kronecker product ($\otimes$) \cite{horn1991topics} to arrive at an exact linear system. Recall that for matrices $X, Y,$ and $Z$ of compatible dimensions:
\begin{equation}
\label{eq:vec_identity}
\vecop(X Y Z) = (X \otimes Z^T) \cdot \vecop(Y)
\end{equation}
Here, $\vecop(Y)$ stacks the rows of matrix $Y$ into a single column vector, and the Kronecker product $X \otimes Z^T$ forms a block matrix where each entry of $X$ multiplies the entire matrix $Z^T$. Applying Equation~\eqref{eq:vec_identity} to Equation~\eqref{eq:commutator} ($I P A - B P I = 0$, for symmetric undirected graphs where $A = A^T$) yields the homogeneous linear system:
\begin{equation}
\label{eq:kronecker_gi}
(I \otimes A - B \otimes I) \cdot \vecop(P) = 0
\end{equation}
where $\vecop(P)$ is an $N^2 \times 1$ column vector representing the vectorized permutation matrix $P$. Note that this re-framing in itself remains identical to the linear graph isomorphism approaches outlined in \cite{tinhofer1986graph,ramana1994fractional} but, importantly, unlocks a richer array of algebraic tools, as we will see in subsequent sections. Equation~\eqref{eq:kronecker_gi} therefore forms the basis of our linear algebraic approach, as it follows directly from Equation~\eqref{eq:kronecker_gi} that all permutation matrices establishing $G_A \simeq G_B$ exist within the nullspace:
\begin{equation*}
\Null(I \otimes A - B \otimes I)
\end{equation*}

\subsection{Subgraph Isomorphism}
\label{sec:sgi_commutator}
Subgraph isomorphism generalizes this question from ``are two graphs the same?'' to ``is this graph a subgraph of the other?''. Formally, for graphs $G_A$ on $N$ vertices and $G_B$ on $M$ (with $M < N$), one is searching for an injective mapping $\pi: V_B \to V_A$ embedding the smaller graph $G_B$ into the larger target $G_A$.
In the linear formulation, this mapping is represented by a rectangular binary matrix $P \in \{0,1\}^{M \times N}$ where $P_{i,v} = 1 \iff \pi(i) = v$. For $P$ to represent a valid one-to-one injection, each row must sum to $1$ and each column to at most $1$:
\begin{equation*}
P \mathbf{1}_N = \mathbf{1}_M \qquad \text{and} \qquad P^T \mathbf{1}_M \le \mathbf{1}_N
\end{equation*}
where $\mathbf{1}_N$ denotes the $N \times 1$ all-ones vector. Whereas traditional integer programming approaches to subgraph matching and graph edit distance introduce $O(M^2 N^2)$ edge-product variables or dummy-vertex padding \cite{justice2006binary,bahiense2012maximum,lerouge2017new}, even though $P$ is now rectangular and no longer orthogonal ($P^T P \neq I_N$), the linear commutator $P A - B P$ from Equation~\eqref{eq:commutator} still directly captures both \textbf{non-induced} subgraph isomorphism ($G_B \subseteq G_A$, where $G_A$ must contain all edges of $G_B$ but may include extra edges) and \textbf{induced} subgraph isomorphism ($G_B \subseteq_{\text{ind}} G_A$, where both edges and non-edges must be preserved, i.e., $G_B \subseteq G_A$ and $\bar{G}_B \subseteq \bar{G}_A$ for the complement graphs $\bar{A} = J_N - I_N - A$ and $\bar{B} = J_M - I_M - B$, with $J_N = \mathbf{1}_N \mathbf{1}_N^T$ the all-ones matrix):
\begin{theorem}[Rectangular Commutator Inequalities for Subgraph Isomorphism]
\label{thm:rect_commutator}
Let $P \in \{0,1\}^{M \times N}$ be a rectangular injection matrix ($P \mathbf{1}_N = \mathbf{1}_M$, $P^T \mathbf{1}_M \le \mathbf{1}_N$). Then:
\begin{enumerate}[label=\arabic*.]
    \item $P$ defines a \textbf{non-induced subgraph isomorphism} ($G_B \subseteq G_A$) if and only if:
    \begin{subequations}
    \label{eq:sgi_bounds}
    \begin{equation}
    \label{eq:sgi_noninduced}
    P A - B P \;\ge\; 0
    \end{equation}
    \item $P$ defines an \textbf{induced subgraph isomorphism} ($G_B \subseteq_{\text{ind}} G_A$) if and only if both $G_B$ and its complement $\bar{G}_B$ satisfy the commutator inequality:
    \begin{equation}
    \label{eq:sgi_induced}
    P A - B P \;\ge\; 0 \qquad \text{and} \qquad P \bar{A} - \bar{B} P \;\ge\; 0
    \end{equation}
    \end{subequations}
\end{enumerate}
\end{theorem}
\begin{proof}
Let pattern vertex $i$ map to target vertex $\pi(i) \in V_A$. Expanding entry $(i,v)$ of the rectangular commutator $P A - B P$ gives:
\begin{equation*}
(P A - B P)_{i,v} = \sum_{u=1}^N P_{i,u} A_{u,v} - \sum_{j=1}^M B_{i,j} P_{j,v} = A_{\pi(i), v} - \sum_{j=1}^M B_{i,j} P_{j,v}
\end{equation*}
Because each column $v$ of $P$ contains at most one non-zero entry, there are two cases:
\begin{itemize}
    \item \textbf{Part 1 (Target vertex $v$ is unmatched, $P_{j,v} = 0\;\forall j$):}
    The second term vanishes and $(P A - B P)_{i,v} = A_{\pi(i),v} \ge 0$ holds automatically (and identically $(P \bar{A} - \bar{B} P)_{i,v} = \bar{A}_{\pi(i),v} \ge 0$). Unmatched target vertices impose no constraints on boundary edges leaving the embedded subgraph.
    \item \textbf{Part 2 (Target vertex $v$ is matched to a unique pattern vertex $j^* = \pi^{-1}(v)$, $P_{j^*,v} = 1$):}
    Here all other entries in column $v$ are zero, so the commutator entry simplifies to:
    \begin{equation*}
    (P A - B P)_{i,v} = A_{\pi(i), \pi(j^*)} - B_{i, j^*} = (P A P^T - B)_{i, j^*} \in \{-1, \, 0, \, +1\}
    \end{equation*}
    Enforcing $P A - B P \ge 0$ forbids $-1$, which holds if and only if $A_{\pi(i),\pi(j^*)} \ge B_{i,j^*}$ for all mapped pairs, establishing non-induced SGI (Equation~\eqref{eq:sgi_noninduced}). By the exact same argument on the complement graphs, enforcing $P \bar{A} - \bar{B} P \ge 0$ holds if and only if $\bar{A}_{\pi(i),\pi(j^*)} \ge \bar{B}_{i,j^*}$ (or equivalently $A_{\pi(i),\pi(j^*)} \le B_{i,j^*}$) for all $i \neq j^*$. Combining both inequalities forces $A_{\pi(i),\pi(j^*)} = B_{i,j^*}$ across all matched pairs, establishing induced SGI (Equation~\eqref{eq:sgi_induced}).
\end{itemize}
\end{proof}

Note that for any binary row-stochastic matrix $P \in \{0,1\}^{M \times N}$ ($P\mathbf{1}_N = \mathbf{1}_M$) and binary adjacency matrices $A, B \in \{0,1\}$, the upper bound $(P A - B P)_{i,v} = A_{\pi(i),v} - (BP)_{i,v} \le A_{\pi(i),v} \le 1$ (and similarly for $\bar{A}, \bar{B}$) holds automatically at every entry $(i,v)$ and serves to tighten continuous LP relaxations. Written out entry-by-entry, each row of $P A - B P \ge 0$ enforces $\sum_{u \in N_A(v)} P_{i,u} \ge \sum_{j \in N_B(i)} P_{j,v}$, recovering the standard linear adjacency-implication inequalities used in 0--1 integer and pseudo-Boolean subgraph encodings \cite{almohamad1993linear,gocht2020subgraph}. Crucially, however, preserving these inequalities in their global matrix-operator form $P A - B P$ rather than treating them as isolated Boolean clauses unlocks a direct algebraic bridge to \textbf{lattice projection and Smith normal form invariant-factor interlacing} across the entire operator (Section~\ref{sec:lattices}). (We also note in passing that the matrix-operator form immediately implies \textbf{polynomial walk monotonicity}: because $P, A, B \ge 0$, $P A \ge B P$ implies $P A^k \ge B^k P$ and hence $P f(A) - f(B) P \ge 0$ for any polynomial $f$ with non-negative coefficients. While all empirical benchmarks in Sections~\ref{sec:ilp_sgi}--\ref{sec:benchmarks_applications} use only the base $k=1$ system below, higher-power walk inequalities offer a natural family of optional auxiliary linear constraints that can be appended without introducing additional decision variables.)

Applying the vectorization identity (Equation~\eqref{eq:vec_identity}) to Equations~\eqref{eq:sgi_noninduced} and \eqref{eq:sgi_induced} yields the $MN \times MN$ Kronecker operator form of subgraph isomorphism:
\begin{subequations}
\label{eq:sgi_kronecker}
\begin{align}
\text{(Non-Induced)} \qquad 0 \;&\le\; (I_M \otimes A - B \otimes I_N) \cdot \vecop(P) \;\le\; 1 \label{eq:sgi_kron_noninduced} \\[4pt]
\text{(Induced, in addition to (\ref{eq:sgi_kron_noninduced}))} \qquad 0 \;&\le\; (I_M \otimes \bar{A} - \bar{B} \otimes I_N) \cdot \vecop(P) \;\le\; 1 \label{eq:sgi_kron_induced}
\end{align}
\end{subequations}

\section{Spectral Properties}
\label{sec:spectral}

\subsection{Eigenvalue Structure}
\label{sec:eigenvalues}

Returning to the Kronecker operator $I \otimes A - B \otimes I$ in Equation~\eqref{eq:kronecker_gi}, its eigenvalues and eigenvectors follow directly from the classical Sylvester--Kronecker spectrum of $A$ and $B$ \cite{horn1991topics,godsil2001algebraic}.

\begin{theorem}[Spectrum of $I \otimes A - B \otimes I$ \cite{horn1991topics}] \label{thm:eigenvalues}
If $\{\lambda_{A,i}\}$ and $\{\lambda_{B,j}\}$ are the eigenvalues of symmetric adjacency matrices $A$ and $B$ with corresponding eigenvectors $v_i$ and $u_j$, then $I \otimes A - B \otimes I$ has eigenvalues $\{\lambda_{A,i} - \lambda_{B,j}\}$ for all $i,j$ with corresponding eigenvectors $w_{i,j} = u_j \otimes v_i$.
\end{theorem}

\begin{proof}
Recall the mixed-product property of Kronecker products: $(P \otimes Q)(R \otimes S) = (PR) \otimes (QS)$. Let $w_{i,j} = u_j \otimes v_i$, where $u_j$ and $v_i$ are the $j$-th and $i$-th eigenvectors of $B$ and $A$, respectively:
\begin{equation*}
(I \otimes A) w_{i,j} = (I \otimes A)(u_j \otimes v_i) = (I u_j) \otimes (A v_i) = u_j \otimes (\lambda_{A,i} v_i) = \lambda_{A,i} w_{i,j}
\end{equation*}
And, by symmetry:
\begin{equation*}
(B \otimes I) w_{i,j} = (B u_j) \otimes (I v_i) = (\lambda_{B,j} u_j) \otimes v_i = \lambda_{B,j} w_{i,j}
\end{equation*}
Because $w_{i,j}$ is an eigenvector of both $(I \otimes A)$ and $(B \otimes I)$, subtracting the two equations yields:
\begin{equation*}
\big((I \otimes A) - (B \otimes I)\big) w_{i,j} = (\lambda_{A,i} - \lambda_{B,j}) w_{i,j}
\end{equation*}
Thus $\{\lambda_{A,i} - \lambda_{B,j}\}$ for all $i,j$ are the exact eigenvalues of the Kronecker commutator.
\end{proof}

An immediate and direct consequence of Theorem~\ref{thm:eigenvalues} is a clean characterization of cospectral graphs \cite{godsil1982constructing,vandam2003which} via nullspace dimensions:

\begin{theorem}[Nullspace Dimension and Cospectrality] 
\label{thm:cospectral}
Let $N_{AB} = \dim \Null(I \otimes A - B \otimes I)$. Then:
\begin{equation}
\label{eq:cospectral}
N_{AA} = N_{AB} = N_{BB} \iff \Spec(G_A) = \Spec(G_B)
\end{equation}
\end{theorem}

That is, if the rank of the Kronecker commutator is consistent across $I \otimes A - A \otimes I$, $I \otimes B - B \otimes I$, and $I \otimes A - B \otimes I$, then $G_A$ is cospectral to $G_B$.

\begin{proof}
Let $\{\mu_1, \dots, \mu_k\}$ denote the union of distinct eigenvalues of the symmetric matrices $A$ and $B$, with multiplicities $m_r \ge 0$ in $\Spec(G_A)$ and $n_r \ge 0$ in $\Spec(G_B)$. Because $A$ and $B$ are orthogonally diagonalizable, $\dim\Null(I \otimes A - B \otimes I)$ equals the number of pairs $(i,j)$ for which $\lambda_{A,i} - \lambda_{B,j} = 0$ by Theorem~\ref{thm:eigenvalues}. Thus $N_{AB} = \sum_{r=1}^k m_r n_r$, and similarly $N_{AA} = \sum_{r=1}^k m_r^2$ and $N_{BB} = \sum_{r=1}^k n_r^2$. Subtracting $2N_{AB}$ from $N_{AA} + N_{BB}$ yields the sum of squares:
\begin{equation*}
N_{AA} - 2 N_{AB} + N_{BB} \;=\; \sum_{r=1}^k (m_r - n_r)^2 \;\ge\; 0.
\end{equation*}
Consequently, $N_{AA} = N_{AB} = N_{BB}$ implies $\sum_{r=1}^k (m_r - n_r)^2 = 0$, which holds if and only if $m_r = n_r$ for every $r \in \{1, \dots, k\}$, i.e., $\Spec(G_A) = \Spec(G_B)$.
\end{proof}

\section{Graph Isomorphism as a Lattice Problem}
\label{sec:gi_lattice_intro}

Though it has been known for some time that lattice isomorphism is at least as hard as graph isomorphism \cite{sikiric2009complexity,haviv2014lattice,ducas2022lattice}, the construction we present in the next section provides an explicit construction---and some novel, surprising outcomes from that construction. Specifically, we show that: there exists an explicit lattice construction for any graph isomorphism problem, graph isomorphism maps to a tightly bound shortest vector problem (SVP), and the Smith normal form allows some graph invariants that allow a powerful set of polynomial-time a priori proofs about (sub)graph isomorphism generally, and that graphs with identical SNF invariants are Generalized Cospectral. We follow with some practical realizations of these properties in some familiar settings, including some explicit linear constructions of permutation matrices directly from the kernels of these systems.

\subsection{Graph Isomorphism and Lattices}
\label{sec:lattices}

Equation~\eqref{eq:kronecker_gi} (and more generally Equation~\eqref{eq:sgi_kronecker}) introduced a natural integer linear programming (ILP) framing over the space of possible permutations that would satisfy forms of graph (sub)isomorphism \cite{berkholz2015limitations,berkholz2017diophantine,grohe2022homomorphism}. Without loss of generality, we first analyze the square graph isomorphism framing of Equation~\eqref{eq:kronecker_gi}:
\[
(I \otimes A - B \otimes I) \cdot \vecop(P) = 0
\]
Which is to say, if $G_A \simeq G_B$, then a $\{0,1\}^{N^2}$ vector exists satisfying the above. Of course, one knows that $P$ is a permutation matrix and a considerably tighter framing is possible. Framed differently, any solution to this equation lies in the intersection of the continuous $\mathbb{R}^{N^2}$ kernel of the Kronecker commutator and the integer lattice defined by all possible permutation bases. The outer convex hull containing the space of these solutions is known as the Birkhoff Polytope \cite{birkhoff1946three}. 

\subsection{Doubly-Stochastic Permutation Basis $L$}
\label{sec:permutation_basis}

The space of all permutation matrices admits a surprisingly clean geometric construction via discrete lattices. 

Recall that a lattice is a discrete additive subgroup of $\mathbb{R}^{N^2}$ generated by integer combinations of basis vectors $y = L z$ ($z \in \mathbb{Z}^d$). In $N \times N$ matrix space ($\dim = N^2$), requiring all $N$ row sums and all $N$ column sums of a matrix $X$ to equal a common scalar $c$ ($X\mathbf{1} = c\mathbf{1}$ and $X^T\mathbf{1} = c\mathbf{1}$) imposes $(N - 1) + (N - 1) = 2N - 2$ independent linear equations. Consequently, the subspace of generalized doubly-stochastic matrices has dimension \cite{birkhoff1946three}:
\begin{equation*}
d \;=\; N^2 - (2N - 2) \;=\; \mathbf{(N - 1)^2 + 1}
\end{equation*}
Rather than enforcing row and column sum equations externally, we can construct an explicit integral basis matrix $L \in \{0,1\}^{N^2 \times d}$ directly from elementary permutation cycles anchored at vertex $0$:
\begin{enumerate}[label=\arabic*.]
    \item \textbf{The Identity ($1$ column):} $\vecop(I_N)$.
    \item \textbf{Primary $2$-Cycles ($N - 1$ columns):} The vectorized star transpositions $\vecop(P_{(0, i)})$ for $i \in \{1, \dots, N-1\}$.
    \item \textbf{Secondary $3$-Cycles ($(N - 1)(N - 2)$ columns):} The vectorized anchored $3$-cycles $\vecop(P_{(0, i, j)})$ for all ordered pairs $1 \le i \neq j \le N-1$.
\end{enumerate}

Stacking these $1 + (N - 1) + (N - 1)(N - 2) = (N - 1)^2 + 1$ elementary cycle permutations forms an explicit unimodular $\{0,1\}$ basis matrix $L \in \{0,1\}^{N^2 \times d}$ for the integer doubly-stochastic lattice:
\begin{equation}
\label{eq:lattice_span}
L(\mathbb{Z}^d) \;=\; \left\{ \vecop(X) \in \mathbb{Z}^{N^2} \;:\; X\mathbf{1} = c\mathbf{1} \;\text{ and }\; X^T\mathbf{1} = c\mathbf{1}, \quad \text{where } c = \mathbf{1}_d^T z = \sum_{k=1}^d z_k \in \mathbb{Z} \right\}
\end{equation}
Because every column of $L$ is a permutation matrix with row and column sum $1$, substituting $\vecop(X) = L z$ ($z \in \mathbb{Z}^d$) into Equation~\eqref{eq:kronecker_gi} automatically enforces equal row and column sums and projects the Kronecker commutator directly onto the permutation lattice:
\begin{equation}
\label{eq:lattice_proj}
(I \otimes A - B \otimes I) L \cdot z = 0  \qquad (z \in \mathbb{Z}^d)       
\end{equation}
or, left-multiplying by $L^T$ to obtain a square symmetric $d \times d$ operator on $\mathbb{Z}^d$:
\begin{equation}
\label{eq:lattice_proj_sym}
L^T (I \otimes A - B \otimes I) L \cdot z = 0 \qquad (z \in \mathbb{Z}^d)       
\end{equation}

Though we'll proceed with Eq \ref{eq:lattice_proj} for the next section, Eq \ref{eq:lattice_proj_sym} will form the basis of relations through the most of rest of this paper we'll often refer to it by short hand $K(A,B)$, though will keep it in its expanded form when readability dictates.

\begin{figure*}[htbp]
    \centering
    \includegraphics[width=\textwidth]{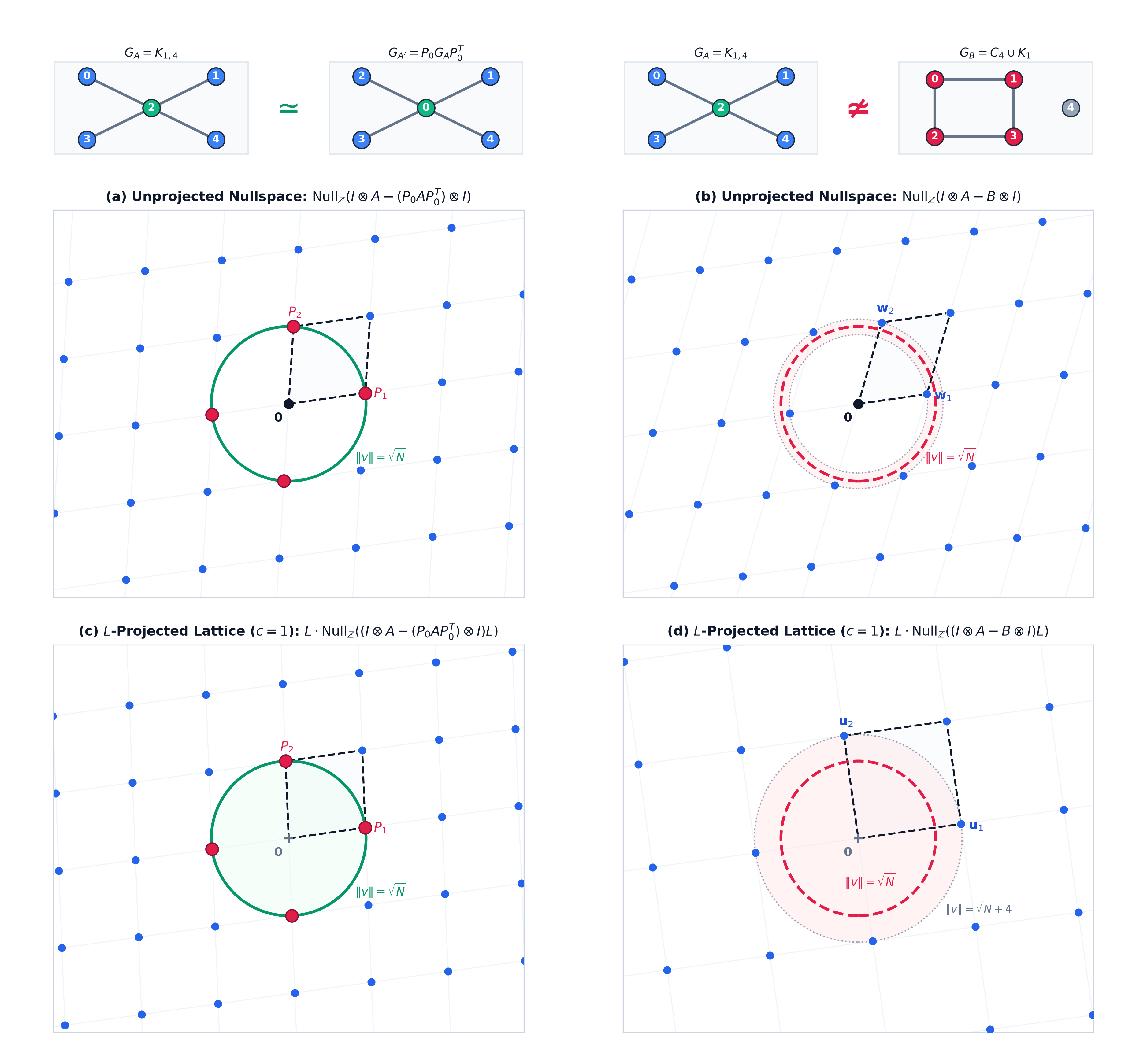}
    \caption{Isometric 2D lattice slices comparing the unprojected integer nullspace $\Null_{\mathbb{Z}}(I \otimes A - B \otimes I)$ (top row) and the Theorem~\ref{thm:shortest_norm} doubly-stochastic lattice $L \cdot \Null_{\mathbb{Z}}((I \otimes A - B \otimes I)L)$ on the unit-sum slice $c = 1$ (bottom row) against the permutation norm sphere $\|v\|_2 = \sqrt{N}$ ($N=5$): (a) unprojected isomorphic slice $\Null_{\mathbb{Z}}(I \otimes A - P_0 A P_0^T \otimes I)$ intersecting $\|v\|_2 = \sqrt{5}$ at $P_1, P_2$; (b) unprojected cospectral slice $\Null_{\mathbb{Z}}(I \otimes A - B \otimes I)$ with interior $c = 0$ vectors $\mathbf{w}_1$ ($\|w_1\|_2 = 2 < \sqrt{5}$) and empty annulus $[\sqrt{4}, \sqrt{6}]$; (c) doubly-stochastic isomorphic slice ($c = 1$), where the interior $\|v\|_2 < \sqrt{5}$ is empty, $P_1, P_2$ lie on $\|v\|_2 = \sqrt{5}$, and non-permutations satisfy $\|v\|_2 \ge \sqrt{N+4}$; (d) doubly-stochastic cospectral slice ($c = 1$), where $c = 0$ interior vectors are eliminated and all lattice points $\mathbf{u}_1, \mathbf{u}_2$ are pushed outside to $\|v\|_2 \ge \sqrt{N+4} = 3$, leaving the entire ball around $\|v\|_2 = \sqrt{5}$ empty.}
    \label{fig:cospectral_lattice}
\end{figure*}

\subsection{Projection onto the Lattice and the Smith Normal Form}
\label{sec:snf_sieve}

Because Equation~\eqref{eq:lattice_proj_sym} is a homogeneous integer linear system over $\mathbb{Z}$ \cite{berkholz2015limitations,berkholz2017diophantine}, we can compute its \textbf{Smith normal form ($\SNF$)} \cite{smith1861systems,newman1972integral,stanley2016smith} in polynomial time \cite{kannan1979polynomial,storjohann1996near}:
\begin{equation}
\label{eq:snf}
U S V = \SNF\!\big((I \otimes A - B \otimes I) L\big)
\end{equation}
where $U$ and $V$ are unimodular integer matrices ($\det U, \det V \in \{\pm 1\}$) and $S = \diag(s_1, \dots, s_r, 0, \dots, 0)$ is diagonal with elementary integer divisors satisfying $s_1 \mid s_2 \mid \dots \mid s_r$. 

\begin{minipage}{\textwidth}

For example, on the smallest cospectral non-isomorphic pair ($N = 5$, Figure~\ref{fig:cospectral_lattice})---the star $G_A = K_{1,4}$ and the $4$-cycle plus isolated vertex $G_B = C_4 \cup K_1$, which share spectrum $\{-2, 0, 0, 0, +2\}$---the unprojected $25 \times 25$ commutator nullspaces are of equal dimension ($N_{AA} = N_{AB} = N_{BB} = 11$ in Theorem~\ref{thm:cospectral}). Projecting onto the $17$-dimensional doubly-stochastic basis $L$, however, immediately separates both their Smith normal form diagonals and their nullspace dimensions (Theorem~\ref{thm:cospectral_lattice}):

\begin{equation*}
\begin{aligned}
\diag(S_{AA}) &= [\,1^6,\; 3,\; \mathbf{0^{10}}\,] \quad &\implies \quad \tilde{N}_{AA} &= \mathbf{10} \quad (L V_0 \in \mathbb{Z}^{25 \times 10}), \\
\diag(S_{BB}) &= [\,1^8,\; 2^3,\; \mathbf{0^6}\,] \quad &\implies \quad \tilde{N}_{BB} &= \mathbf{6} \quad (L V_0 \in \mathbb{Z}^{25 \times 6}), \\
\diag(S_{AB}) &= [\,1^{11},\; \mathbf{0^6}\,] \quad &\implies \quad \tilde{N}_{AB} &= \mathbf{6} \quad (L V_0 \in \mathbb{Z}^{25 \times 6}).
\end{aligned}
\end{equation*}

Note that the SNF acts as an integer analog of Singular Value Decomposition and that the $0$ valued invariants form a basis for integer valued solutions to \ref{eq:snf}. Specifically, if $V_0 \in \mathbb{Z}^{d \times k}$ denotes the columns of $V$ corresponding to zeros on the diagonal of $S$, mapping $V_0$ back to the ambient space via $L$ yields the exact integral generator for the intersection of the commutator nullspace with the doubly-stochastic lattice:
\begin{equation}
\label{eq:rho_gen}
\vecop(P) \;=\; (L \cdot V_0) \, w \qquad (L \cdot V_0 \in \mathbb{Z}^{N^2 \times k}, \;\; w \in \mathbb{Z}^k)
\end{equation}

Having a look at the basis vectors that result from Eq \ref{eq:rho_gen} for our classic cospectral example is illuminating

\centering
\includegraphics[width=0.9\textwidth]{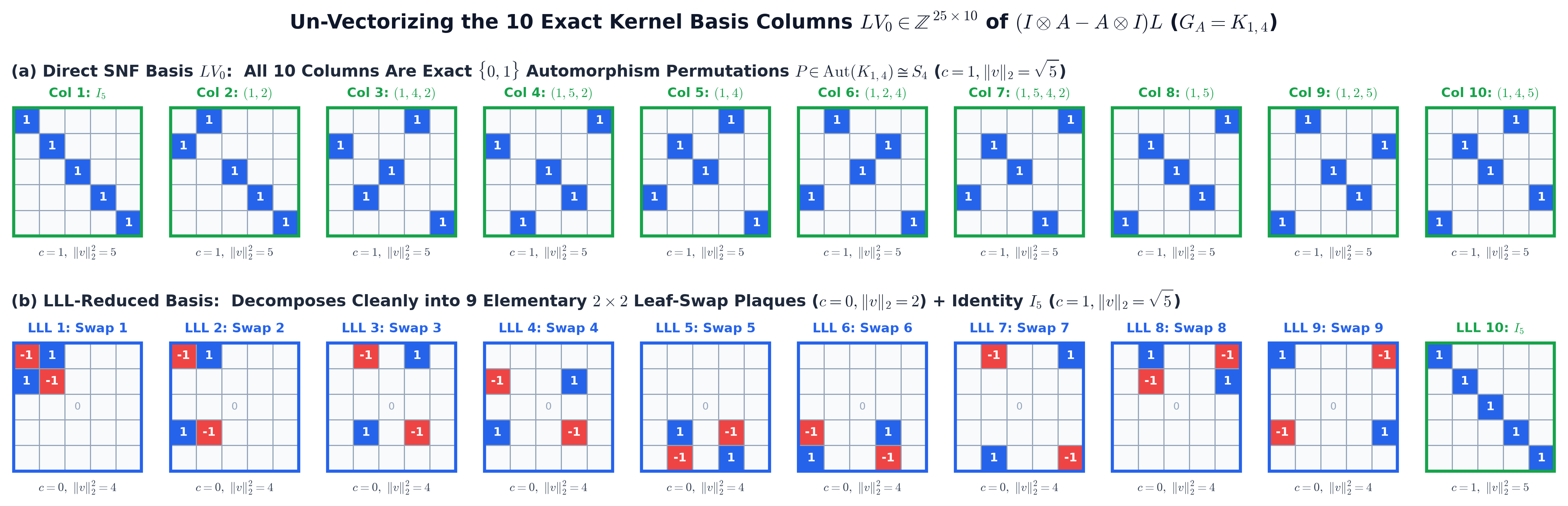}
\end{minipage}

Note that $10$ of the $11$ raw basis vectors in $L V_0$ resulting directly from the $\SNF$ are actual automorphism permutation matrices of $K_{1,4}$, which affine-linearly generate all $|\Aut(K_{1,4})| = 4! = 24$ automorphisms. Moreover, running the Lenstra--Lenstra--Lov\'asz (LLL) lattice reduction algorithm \cite{lenstra1982factoring} cleanly isolates $I_5$ and the nine elementary $2 \times 2$ leaf swaps. Though $|\Aut(G)|$ can obviously grow much faster than the lattice dimension $\le N^2$ for larger graphs, this example gives a concrete, visual picture of how the integer commutator lattice intersects the Birkhoff polytope.

It is worth pausing to reflect that Equation~\eqref{eq:rho_gen} offers a fully realized, computable generator for doubly-stochastic, integer solutions to the intertwiner (Equation~\eqref{eq:commutator}) for graphs $G_A$ and $G_B$. Moreover, though the lattice isomorphism problem (LIP) is known to be at least as hard as graph isomorphism ($\text{GI} \le_P \text{LIP}$) \cite{sikiric2009complexity,haviv2014lattice,ducas2022lattice}, the above method demonstrates an explicit, constructive mapping from $\text{GI}$ to integer lattices. In the algebraic graph theory literature, the Smith normal form has previously been applied to adjacency or Laplacian matrices directly \cite{rushanan1995eigenvalues,stanley2016smith}, to unprojected linear Diophantine systems \cite{berkholz2015limitations,berkholz2017diophantine}, and to the $N \times N$ walk matrix $W(A) = [\mathbf{1}, A\mathbf{1}, \dots, A^{N-1}\mathbf{1}]$ to study generalized cospectrality \cite{johnson1980note,wang2006sufficient,abiad2012cospectral,wang2019cospectral,abiad2022codeterminantal,qiu2023smith}. However, walk-matrix methods require the graph to be controllable ($\det W(A) \neq 0$) and degenerate on regular graphs where $A\mathbf{1} = k\mathbf{1}$ collapses $W(A)$ to rank~$1$. Projecting the $N^2 \times N^2$ Kronecker commutator directly onto the cycle basis $L$ bypasses walk vectors entirely---operating on arbitrary (including regular) graphs, capturing both $G$ and its complement $\bar{G}$ in a single operator, and extending naturally to rectangular subgraph isomorphism ($M < N$, Section~\ref{sec:snf_sgi}).

Though this framing provides a rich theoretical landscape for investigation, three useful theorems present themselves immediately:

\begin{theorem}[Graph Isomorphism is a Shortest Vector Problem]
\label{thm:shortest_norm}
Let $z \in \mathbb{Z}^d$ satisfy Equation~\eqref{eq:lattice_proj} (or Equation~\eqref{eq:lattice_proj_sym}) on the unit-sum affine slice $\mathbf{1}_d^T z = 1$, with ambient matrix $v = L z = \vecop(X) \in \mathbb{Z}^{N^2}$. Then:
\begin{gather*}
    \|L z\|_2^2 = N \quad (\text{with } \mathbf{1}_d^T z = 1) \iff X \in S_N \;\text{ and }\; G_A \simeq G_B, \\
    \text{whereas every non-permutation solution with } \mathbf{1}_d^T z = 1 \text{ satisfies } \|L z\|_2^2 \ge N + 2.
\end{gather*}
\end{theorem}
\begin{proof}
For any integer matrix $X \in \mathbb{Z}^{N \times N}$ with row and column sums $c = \mathbf{1}_d^T z = 1$, we have $\sum_{i,j} X_{i,j} = N$. Since $x(x-1) \ge 0$ holds for every integer $x \in \mathbb{Z}$ with equality if and only if $x \in \{0,1\}$, and $x(x-1) \ge 2$ is even whenever $x \in \mathbb{Z} \setminus \{0,1\}$, summing over all entries yields:
\begin{equation*}
\|L z\|_2^2 \;=\; \sum_{i,j=1}^N X_{i,j}^2 \;=\; N + \sum_{i,j=1}^N X_{i,j}(X_{i,j} - 1) \;\ge\; N,
\end{equation*}
with equality $\|L z\|_2^2 = N \iff X \in \{0,1\}^{N \times N}$ (which with $X\mathbf{1}=\mathbf{1}$ and $X^T\mathbf{1}=\mathbf{1}$ means $X \in S_N$), and $\|L z\|_2^2 \ge N + 2$ whenever $X \notin S_N$. Finally, for $X = P \in S_N$, left-projection by $L^T$ is lossless because $P^T(PA - BP) = A - P^T B P$ is symmetric with zero diagonal, and the only symmetric matrix with zero diagonal in $(I_N \otimes P^T)\operatorname{Col}(L)^\perp = \{r\mathbf{1}^T + \mathbf{1}c^T : \mathbf{1}^T r = \mathbf{1}^T c\}$ is $0$; hence $L^T\vecop(PA - BP) = 0 \iff PA - BP = 0 \iff G_A \simeq G_B$.
\end{proof}

\begin{theorem}[Smith Normal Form Invariants and Generalized Cospectrality] 
\label{thm:cospectral_lattice}
Let $\tilde{N}_{AB} = \dim \Null\!\big((I \otimes A - B \otimes I) L\big)$. For graphs $G_A, G_B$ satisfying $N_{AA} = N_{AB} = N_{BB}$ in Theorem~\ref{thm:cospectral} \cite{johnson1980note,wang2006sufficient}:
\begin{equation}
\label{eq:cospectral_lattice}
\tilde{N}_{AA} = \tilde{N}_{AB} = \tilde{N}_{BB} \iff \Spec(G_A) = \Spec(G_B) \quad \text{and} \quad \Spec(\bar{G}_A) = \Spec(\bar{G}_B).
\end{equation}
\end{theorem}
\begin{proof}
Let $P_r(A)$ and $P_r(B)$ denote the orthogonal projectors onto the $\mu_r$-eigenspaces $E_r(A)$ and $E_r(B)$, and set $v_r(A) = P_r(A)\mathbf{1}$ and $v_r(B) = P_r(B)\mathbf{1}$. Any $\vecop(X) = L z$ in $\Null((I \otimes A - B \otimes I)L)$ satisfies $X A = B X$ (so $X P_r(A) = P_r(B) X$) together with $X\mathbf{1} = c\mathbf{1}$ and $X^T\mathbf{1} = c\mathbf{1}$, which implies $X v_r(A) = c v_r(B)$ and $X^T v_r(B) = c v_r(A)$ for every $r$. On the $c = 0$ slice, $X$ maps $E_r(A) \cap v_r(A)^\perp$ into $E_r(B) \cap v_r(B)^\perp$, contributing $\tilde{N}_{AB}^{(0)} = \sum_r (m_r - \delta_{v_r(A) \neq 0})(n_r - \delta_{v_r(B) \neq 0})$ dimensions. Furthermore, taking the inner product with $v_r(B)$ yields
\[
c\,\|v_r(B)\|_2^2 \;=\; v_r(B)^T X v_r(A) \;=\; (X^T v_r(B))^T v_r(A) \;=\; c\,\|v_r(A)\|_2^2 \qquad \forall\, r.
\]
Thus a $c \neq 0$ solution exists (adding $+1$ to $\tilde{N}_{AB}^{(0)}$ as in $\tilde{N}_{AA}$ and $\tilde{N}_{BB}$) if and only if $\|v_r(A)\|_2 = \|v_r(B)\|_2$ for all $r$. Together with $m_r = n_r$ from $N_{AA} = N_{AB} = N_{BB}$, equality of the main projections $\|v_r(A)\|_2 = \|v_r(B)\|_2$ across all eigenspaces is equivalent by the Johnson--Newman theorem \cite{johnson1980note} to $\Spec(G_A) = \Spec(G_B)$ and $\Spec(\bar{G}_A) = \Spec(\bar{G}_B)$.
\end{proof}

\begin{theorem}[Rational Denominators and Invariant Factors]
\label{thm:level_bound}
Let $M_0(A,B) = (I_N \otimes A - B \otimes I_N) L_{N,N} \in \mathbb{Z}^{N^2 \times (N-1)^2}$ denote the commutator restricted to the $(N-1)^2$-dimensional zero-sum difference lattice $\mathcal{L}^{(0)}_{N,N}$ (generated by $L_{N,N}$ in Section~\ref{sec:snf_sgi}), with non-zero Smith normal form invariant factors $s_1 \mid s_2 \mid \dots \mid s_{\max}$ \cite{rushanan1995eigenvalues,stanley2016smith}. Any regular rational orthogonal matrix $Q \in O(N, \mathbb{Q})$ ($Q A - B Q = 0$, $Q \mathbf{1} = \mathbf{1}$) of level $\ell = \min\{m \in \mathbb{Z}_{\ge 1} : m Q \in \mathbb{Z}^{N \times N}\}$ \cite{wang2006sufficient} decomposes around a reference permutation $P_0 \in S_N$ as $Q = P_0 + X_0$, where $\vecop(X_0) = L_{N,N} z_0$ ($z_0 \in \mathbb{Q}^{(N-1)^2}$) satisfies the inhomogeneous integer system $M_0(A,B) z_0 = -\vecop(P_0 A - B P_0)$. Whenever $M_0(A,B)$ has full column rank (or on the normal component $z_0 \perp \Null(M_0(A,B))$), the level $\ell$ satisfies:
\begin{equation}
    \ell \;\Big|\; s_{\max}.
\end{equation}
Consequently:
\begin{enumerate}
    \item \textbf{Trivial Torsion Certifies Isomorphism:} If $s_{\max} = 1$, then $\ell = 1$, forcing $Q \in O(N, \mathbb{Z})$ with $Q \mathbf{1} = \mathbf{1} \implies G_A \simeq G_B$.
    \item \textbf{Immunity to Fractional Switching:} If $\gcd(s_{\max}, p) = 1$, the pair cannot be related by any such rational switching operation of level $p$. In particular, if $s_{\max}$ is odd, no full-rank cospectral mate can be generated by standard Godsil--McKay switching ($\ell = 2$) \cite{godsil1982constructing}.
\end{enumerate}
\end{theorem}
\begin{proof}
Because $Q\mathbf{1} = \mathbf{1} = P_0\mathbf{1}$ and $Q^T\mathbf{1} = \mathbf{1} = P_0^T\mathbf{1}$, the difference $X_0 = Q - P_0$ has zero row and column sums, so $\vecop(X_0) = L_{N,N} z_0$ for a unique coordinate vector $z_0 \in \mathbb{Q}^{(N-1)^2}$ whose least common denominator equals $\ell$ (since $L_{N,N}$ is unimodular on $\mathcal{L}^{(0)}_{N,N}$). Substituting $Q = P_0 + X_0$ into $QA - BQ = 0$ yields $M_0(A,B) z_0 = b$, where $b = -\vecop(P_0 A - B P_0) \in \mathbb{Z}^{N^2}$. Writing the rectangular Smith normal form $U M_0(A,B) V = S$ with unimodular $U \in \mathrm{GL}(N^2, \mathbb{Z})$ and $V \in \mathrm{GL}((N-1)^2, \mathbb{Z})$, the top $(N-1)^2$ rows give $\diag(s_1, \dots, s_{\max}) (V^{-1} z_0) = (U b)_{1:(N-1)^2} \in \mathbb{Z}^{(N-1)^2}$. Since every $s_i$ divides $s_{\max}$ and $V$ is unimodular over $\mathbb{Z}$, it follows that $s_{\max} z_0 \in \mathbb{Z}^{(N-1)^2}$, and therefore $\ell \mid s_{\max}$.
\end{proof}

\subsection{The Smith Normal Form and Subgraph Isomorphism}
\label{sec:snf_sgi}

A natural question  is how the Smith normal form might relate to the more general problem of subgraph isomorphism. Surprisingly, the formulation maps naturally so long as the problem is framed in the context of the symmetric lattice commutator
\begin{equation*}
K(A,B) = L^T (I \otimes A - B \otimes I) L
\end{equation*}
 When $M < N$, moving from square permutations $P \in \{0,1\}^{N \times N}$ to rectangular injections $P \in \{0,1\}^{M \times N}$ introduces two obstacles that are absent when $M = N$:
\begin{enumerate}[leftmargin=1.8em]
    \item \textbf{Asymmetric Column Sums Under Vertex Relabeling:} When $M = N$, every permutation matrix has column sums $(1, 1, \dots, 1)$, which is invariant under any relabeling of $G_A$'s vertices. When $M < N$, however, a rectangular injection $P \in \{0,1\}^{M \times N}$ matches $M$ columns of $G_A$ (column sum $1$) and leaves $N - M$ columns unmatched (column sum $0$). Including a reference injection such as $P_0 = \begin{bmatrix} I_M & \mathbf{0}_{M \times (N-M)} \end{bmatrix}$ in the lattice basis would uniquely identify the first $M$ vertices of $G_A$, meaning that any relabeling of the vertices of $G_A$ would permute the $1$s and $0$s across columns and alter the resulting Smith normal form.
    \item \textbf{Boundary Edges Leaving the Subgraph:} Even when $G_B$ is an induced subgraph of $G_A$, edges in $G_A$ connecting the $M$ matched vertices to the $N - M$ unmatched vertices create non-zero $+1$ entries in the unmatched columns of $P A - B P$.
\end{enumerate}
Both obstacles disappear if we consider \textbf{differences between candidate injections} rather than individual injection matrices.

If $P_1, P_2 \in \{0,1\}^{M \times N}$ are any two rectangular injections mapping $V_B$ onto the \emph{same} $M$-vertex subset of $G_A$, their difference matrix $X = P_1 - P_2$ exhibits two exact cancellations:
\begin{enumerate}[label=(\roman*),leftmargin=1.8em]
    \item \textbf{Row sums vanish:} Because each row of $P_1$ and $P_2$ sums to $1$ (every vertex of $G_B$ is mapped exactly once), every row of $X$ sums to $1 - 1 = 0$, giving $X \mathbf{1}_N = \mathbf{0}_M$.
    \item \textbf{Column sums vanish universally:} On the $M$ matched columns, both matrices have column sum $1$, so $1 - 1 = 0$. On the $N - M$ unmatched columns, neither matrix has any entries, so $0 - 0 = 0$. Thus \emph{every} column sum of $X$ is identically zero, giving $X^T \mathbf{1}_M = \mathbf{0}_N$.
\end{enumerate}
This motivates restricting our algebraic analysis to the \textbf{zero-sum difference lattice}:
\begin{equation}
\label{eq:zero_sum_lattice}
\mathcal{L}^{(0)}_{M,N} = \Big\{\, X \in \mathbb{Z}^{M \times N} \;:\; X \mathbf{1}_N = \mathbf{0}_M \;\text{ and }\; X^T \mathbf{1}_M = \mathbf{0}_N \,\Big\}.
\end{equation}
Working in $\mathcal{L}^{(0)}_{M,N}$ neatly eliminates both obstacles:
\begin{itemize}[leftmargin=1.5em]
    \item \textbf{Subset and Relabeling Independence:} Because every column sum in $\mathcal{L}^{(0)}_{M,N}$ is $0$ regardless of which columns are matched or unmatched, $\mathcal{L}^{(0)}_{M,N}$ does not depend on which $M$ vertices of $G_A$ are chosen. Permuting the vertices of $G_B$ or $G_A$ ($X \mapsto P_B X P_A^T$) simply permutes the rows and columns of $X$, mapping $\mathcal{L}^{(0)}_{M,N}$ bijectively onto itself via a unimodular transformation in $\mathrm{GL}(d_{M,N}, \mathbb{Z})$.
    \item \textbf{Boundary Edge Annihilation:} Because the row and column sums of $X$ vanish identically, projecting the commutator operator onto $\mathcal{L}^{(0)}_{M,N}$ cancels the contributions of boundary edges exiting the subgraph, isolating the internal commutator structure of the induced copy.
\end{itemize}

Geometrically, $\mathcal{L}^{(0)}_{M,N}$ is the $(M-1)(N-1)$-dimensional lattice of integer swap moves (the tangent space of zero-sum perturbations to doubly-stochastic assignments), generated unimodularly by the $(M-1)(N-1)$ elementary pairwise swap differences:
\begin{equation}
\label{eq:cycle_basis}
C_{i,j} \;=\; E_{1,1} - E_{1,j} - E_{i,1} + E_{i,j} \;\in\; \{-1,0,1\}^{M \times N}, \qquad 2 \le i \le M,\;\; 2 \le j \le N.
\end{equation}
Writing $L_{M,N} \in \{-1,0,1\}^{MN \times (M-1)(N-1)}$ for the matrix whose columns are $\vecop(C_{i,j})$, define the projected commutator matrices:
\begin{equation}
\label{eq:projected_kab_kbb}
K(A,B) \;=\; L_{M,N}^T (I_M \otimes A - B \otimes I_N) L_{M,N}, \qquad
K(B,B) \;=\; L_{M,M}^T (I_M \otimes B - B \otimes I_M) L_{M,M}
\end{equation}
of dimensions $d_{M,N} = (M-1)(N-1)$ and $d_{M,M} = (M-1)^2$, separated by the dimension gap:
\begin{equation}
\label{eq:delta_gap}
\Delta = d_{M,N} - d_{M,M} \;=\; (M-1)(N-1) - (M-1)^2 \;=\; (N - M)(M - 1).
\end{equation}
\begin{theorem}[SNF Interlacing for Induced Subgraphs]
\label{thm:snf_interlacing}
Let $G_A$ ($N$ vertices) and $G_B$ ($M \le N$ vertices) be graphs with projected commutator matrices $K(A,B)$ and $K(B,B)$ as in Equation~\eqref{eq:projected_kab_kbb}. Let
\[
s_{BB} = (s_{BB,1}, \dots, s_{BB,d_{M,M}}) \quad\text{and}\quad s_{AB} = (s_{AB,1}, \dots, s_{AB,d_{M,N}})
\]
denote the diagonal invariant factors of $\SNF(K(B,B))$ and $\SNF(K(A,B))$ ordered with non-zero divisors first ($s_1 \mid s_2 \mid \dots \mid s_r$) followed by zeros. 

If $G_B \subseteq_{\mathrm{ind}} G_A$, then for every $1 \le i \le (M-1)^2$, the Smith normal form invariants satisfy the divisibility sandwich:
\begin{equation}
\label{eq:snf_sandwich}
s_{AB, \, i} \;\Big|\; s_{BB, \, i} \;\Big|\; s_{AB, \, i + 2(N-M)(M-1)}
\end{equation}
and the nullspace dimensions satisfy the lower bound:
\begin{equation}
\label{eq:nullity_sgi_bound}
\dim\Null(K(A,B)) \;\ge\; \dim\Null(K(B,B)) - (N-M)(M-1).
\end{equation}
\end{theorem}
\begin{proof}
Suppose $G_B \subseteq_{\mathrm{ind}} G_A$. By relabeling the vertices of $G_A$ via a permutation matrix $P_A \in S_N$ (which acts on $L_{M,N}$ via a unimodular basis matrix $U \in \mathrm{GL}(d_{M,N}, \mathbb{Z})$ and leaves $\SNF(K(A,B))$ invariant), we may place the induced copy of $G_B$ on the first $M$ vertices of $G_A$:
\[
P_A A P_A^T \;=\; \begin{bmatrix} B & C \\ C^T & D \end{bmatrix},
\]
where $C \in \{0,1\}^{M \times (N-M)}$ encodes the boundary edges between $G_B$ and $V_A \setminus V_B$. Partition the basis $L_{M,N} = \begin{bmatrix} \tilde{L}_{M,M} & L_{\mathrm{ext}} \end{bmatrix}$ into the $d_{M,M} = (M-1)^2$ cycle matrices $C_{i,j}$ ($2 \le i,j \le M$) supported on the first $M$ columns---which are precisely the basis matrices of $L_{M,M}$ padded with $N-M$ zero columns, $\tilde{X} = \begin{bmatrix} X & 0 \end{bmatrix}$---and the remaining $\Delta = (N-M)(M-1)$ cycle matrices ($M < j \le N$). For any two padded matrices $\tilde{X} = \begin{bmatrix} X & 0 \end{bmatrix}$ and $\tilde{Y} = \begin{bmatrix} Y & 0 \end{bmatrix}$ in $\tilde{L}_{M,M}$, the commutator evaluates to:
\[
\tilde{X} \begin{bmatrix} B & C \\ C^T & D \end{bmatrix} - B \tilde{X}
\;=\; \begin{bmatrix} X B - B X & X C \end{bmatrix}.
\]
When we project against $\tilde{Y}$ on the left, the $N - M$ zero columns of $\tilde{Y} = \begin{bmatrix} Y & 0 \end{bmatrix}$ multiply the boundary block $X C$ by zero, annihilating all boundary edges:
\[
\vecop(\tilde{Y})^T (I_M \otimes P_A A P_A^T - B \otimes I_N) \vecop(\tilde{X})
\;=\; \operatorname{tr}\!\left(\begin{bmatrix} Y^T \\ 0 \end{bmatrix} \begin{bmatrix} X B - B X & X C \end{bmatrix}\right)
\;=\; \operatorname{tr}\!\big(Y^T (X B - B X)\big).
\]
Thus the top-left $d_{M,M} \times d_{M,M}$ principal submatrix of $U^T K(A,B) U$ is identically $K(B,B)$. Because $K(B,B)$ is obtained from $U^T K(A,B) U$ by deleting $\Delta = (N-M)(M-1)$ rows and $\Delta = (N-M)(M-1)$ columns, the classical Smith normal form submatrix interlacing theorem \cite{thompson1979smith} guarantees $s_{AB,i} \mid s_{BB,i} \mid s_{AB,i + 2\Delta}$ for all $1 \le i \le d_{M,M}$, and Cauchy rank interlacing yields Equation~\eqref{eq:nullity_sgi_bound}.
\end{proof}

\begin{theorem}[SNF Interlacing for Non-Induced Subgraphs]
\label{thm:snf_noninduced}
Let $G_A$ ($N$ vertices) and $G_B$ ($M \le N$ vertices) be graphs where every vertex of $G_B$ has degree at least $d$. If no two vertices of degree $> d$ in $G_A$ share an edge, then $G_B \subseteq G_A$ satisfies Equations~\eqref{eq:snf_sandwich} and Equation~\eqref{eq:nullity_sgi_bound} of Theorem~\ref{thm:snf_interlacing}.
\end{theorem}
\begin{proof}
Suppose $G_B \subseteq G_A$. Each matched vertex in $G_A$ already uses at least $d$ edges to match the edges of $G_B$. Any extra edge in $G_A$ between two matched vertices would increase both of their degrees in $G_A$ to at least $d + 1 > d$. Since no two vertices of degree $> d$ in $G_A$ share an edge, no such extra edge can exist, so $G_B \subseteq_{\mathrm{ind}} G_A$ and Theorem~\ref{thm:snf_interlacing} applies.
\end{proof}

\paragraph{Example 1: Proving Planarity from the SNF}

Since Theorem~\ref{thm:snf_interlacing} is a necessary condition for subgraph embedding, refuting a forbidden pattern acts as a polynomial-time \emph{non-containment certificate}. By Kuratowski's theorem \cite{kuratowski1930probleme}, a graph is planar if and only if it contains no subdivision of $K_5$ or $K_{3,3}$, which is to say that rejecting $K_5$ and $K_{3,3}$ (along with candidate subdivisions on $\le N$ vertices) via $\SNF$ interlacing immediately guarantees that the graph in question is planar. 

By way of example, let the target graph be the $7$-vertex planar Pentagonal Bipyramid $G_A = J_{13}$ ($N = 7, |E_A| = 15$, Figure~\ref{fig:snf_planar_polyhedron}a), formed by joining two polar hubs $\{1,2\}$ to an equatorial $5$-cycle $\{3,4,5,6,7\}$. Though a planar embedding exists, it is a classical visual example of a ``difficult'' toy example for planarity because the planar drawing is not visually obvious (Figure~\ref{fig:snf_planar_polyhedron}):
\begin{itemize}[leftmargin=1.8em]
    \item \textbf{Case 1 (Kuratowski Clique $G_{B'} = K_5 \not\subseteq J_{13}$, Refuted):}
    Here $M = 5$, $d_{5,5} = 16$, $d_{5,7} = 24$, and $2\Delta = 16$. Because $K_5$ is a complete graph (so induced and non-induced containment coincide) and commutes with every $5 \times 5$ permutation matrix, $K(B',B') = 0$:
    \[
    \begin{aligned}
    s_{B'B'} &= [\,\underbrace{\mathbf{0,\; 0,\; 0,\; 0,\; 0,\; 0,\; 0,\; 0}}_{i=1..8},\; \underbrace{0,\; \dots,\; 0}_{i=9..16}\,] \quad (\text{length } 16) \\
    s_{AB'}  &= [\,\underbrace{1,\; \dots,\; 1}_{1..16},\; \underbrace{\mathbf{1,\; 1,\; 5,\; 5,\; 45,\; 45,\; 45,\; 45}}_{i+16=17..24}\,] \quad (\text{length } 24)
    \end{aligned}
    \]
    At $i = 1, \dots, 8$, we have $s_{B'B',i} = 0$, whereas at $i + 16 = 17, \dots, 24$ we have $s_{AB',i+16} \in \{1, 5, 45\}$. Because $\mathbf{0 \nmid 1}$, $\mathbf{0 \nmid 5}$, and $\mathbf{0 \nmid 45}$, Theorem~\ref{thm:snf_interlacing} immediately refutes $K_5 \subseteq J_{13}$ (and the nullity bound~\eqref{eq:nullity_sgi_bound} also fails: $0 < 16 - 8 = 8$). In fact, because $s_{K_5,K_5} = [0^{16}]$ universally, \emph{any} $N$-vertex graph containing $K_5$ must satisfy $\dim\Null(K(A,K_5)) \ge 16 - 4(N-5)$.

    \item \textbf{Case 2 (Utility graph $G_{B''} = K_{3,3} \not\subseteq_{\mathrm{ind}} J_{13}$, Refuted):}
    Here $M = 6$, $d_{6,6} = 25$, $d_{6,7} = 30$, $\Delta = (7-6)(6-1) = 5$, and the shift is $2\Delta = 10$:
    \[
    \begin{aligned}
    s_{B''B''} &= [\,\underbrace{\mathbf{2,\; 2,\; 2,\; 2,\; 2,\; 2}}_{i=1..6},\; 6,\; 6,\; \underbrace{\mathbf{0,\; \dots,\; 0}}_{i=9..16},\; \underbrace{0,\; \dots,\; 0}_{17..25}\,] \quad (\text{length } 25) \\
    s_{AB''}   &= [\,\underbrace{1,\; \dots,\; 1}_{1..10},\; \underbrace{\mathbf{1,\; 1,\; 1,\; 1,\; 1,\; 1}}_{i+10=11..16},\; 2,\; 2,\; \underbrace{\mathbf{2,\; 2,\; 2,\; 6,\; 24,\; 24,\; 120,\; 600}}_{i+10=19..26},\; \underbrace{0,\; \dots,\; 0}_{27..30}\,] \quad (\text{length } 30)
    \end{aligned}
    \]
    At $i = 1, \dots, 6$, we have $s_{B''B'',i} = 2$, whereas at $i + 10 = 11, \dots, 16$ we have $s_{AB'',i+10} = 1$. Because $\mathbf{2 \nmid 1}$ (and at $i = 9, \dots, 16$, $\mathbf{0 \nmid s_{AB'',19..26}}$), Theorem~\ref{thm:snf_interlacing} immediately refutes $K_{3,3} \subseteq_{\mathrm{ind}} J_{13}$ (and the nullity bound~\eqref{eq:nullity_sgi_bound} also fails: $4 < 17 - 5 = 12$).

    Because we can guarantee that neither graph is present we have an immediate guarantee that the graph is planar without further investigation.
\end{itemize}

\begin{figure}[htbp]
    \centering
    \includegraphics[width=\textwidth]{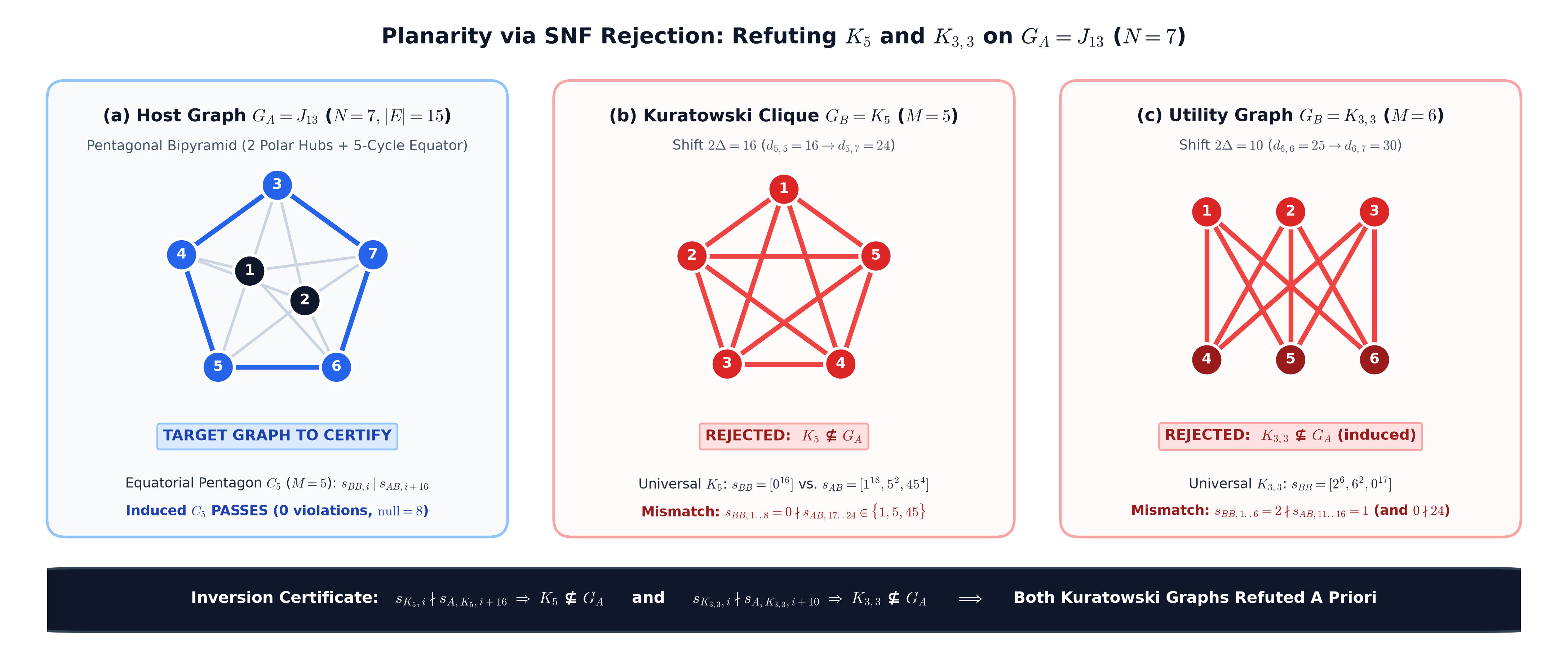}
    \caption{Inversion certificate on the $N = 7$ planar Pentagonal Bipyramid $G_A = J_{13}$ ($|E_A| = 15$): (a) the equatorial pentagon $C_5$ ($M=5$, blue) passes Theorem~\ref{thm:snf_interlacing} with zero violations; (b) the Kuratowski clique $K_5$ ($M=5$) has universal signature $s_{B'B'} = [0^{16}]$ and is rejected via $0 \nmid s_{AB',17..24}$; (c) the utility graph $K_{3,3}$ ($M=6$) has universal signature $s_{B''B''} = [2^6, 6^2, 0^{17}]$ and is rejected via $s_{B''B'',1..6} = 2 \nmid s_{AB'',11..16} = 1$.}
    \label{fig:snf_planar_polyhedron}
\end{figure}

\paragraph{Example 2: Maximum Cycle Length}
Another classical existence problem in subgraph isomorphism is the determination of Hamiltonian (closed loop) cycles in a target graph. While the general problem is NP-complete, the $\SNF$ framing here allows quick certificates of non-existence of cycles of arbitrary size. Note that a negative result proves non-existence, while a positive result only allows for the possibility of existence, requiring further investigation to fully prove.

For example, let the target graph be the $2$-connected, $7$-vertex Theta-Bridge graph $G_A$ ($N = 7, |E_A| = 9$, Figure~\ref{fig:snf_small_hamiltonian}a), formed by joining two degree-$4$ hubs $\{3,6\}$ via three paths of length~$2$ and one path of length~$3$. Because the only two vertices of degree $>2$ are the non-adjacent hubs $\{3,6\}$, Theorem~\ref{thm:snf_noninduced} applies directly to every cycle $C_M \subseteq G_A$, and testing the cycle ladder $G_B = C_5, C_6, C_7$ determines the largest cycle that can possibly exist inside $G_A$ a priori (Figure~\ref{fig:snf_small_hamiltonian}):
\begin{itemize}[leftmargin=1.8em]
    \item \textbf{Case 1 ($5$-Cycle $G_B = C_5 \subseteq_{\mathrm{ind}} G_A$, Passes):}
    Here $M = 5$, $d_{5,5} = 16$, $d_{5,7} = 24$, $\Delta = 8$, and the shift is $2\Delta = 16$:
    \[
    \begin{aligned}
    s_{BB} &= [\,\underbrace{\mathbf{1,\; 1,\; 1,\; 1}}_{i=1..4},\; \underbrace{\mathbf{5,\; 5,\; 5,\; 5}}_{i=5..8},\; \underbrace{0,\; \dots,\; 0}_{i=9..16}\,] \quad (\text{length } 16) \\
    s_{AB} &= [\,\underbrace{1,\; \dots,\; 1}_{1..13},\; 5,\; 5,\; 5,\; \underbrace{\mathbf{10,\; 10,\; 10,\; 10}}_{i+16=17..20},\; \underbrace{\mathbf{0,\; 0,\; 0,\; 0}}_{i+16=21..24}\,] \quad (\text{length } 24)
    \end{aligned}
    \]
    For $i = 1, \dots, 4$, $s_{AB,i} = 1 \mid s_{BB,i} = 1 \mid s_{AB,i+16} = 10$, and for $i = 5, \dots, 8$, $s_{BB,i} = 5 \mid s_{AB,i+16} = 0$. The nullity bound~\eqref{eq:nullity_sgi_bound} also holds ($4 \ge 8 - 8 = 0$), matching the green $5$-cycle $(3\text{--}5\text{--}6\text{--}7\text{--}4\text{--}3)$ in Figure~\ref{fig:snf_small_hamiltonian}a. The possibility of a $C_5$ has not been rejected.

    \item \textbf{Case 2 ($6$-Cycle $G_{B'} = C_6 \not\subseteq_{\mathrm{ind}} G_A$, Refuted):}
    Here $M = 6$, $d_{6,6} = 25$, $d_{6,7} = 30$, $\Delta = 5$, and the shift is $2\Delta = 10$:
    \[
    \begin{aligned}
    s_{B'B'} &= [\,\underbrace{1,\; \dots,\; 1}_{1..8},\; \underbrace{\mathbf{2,\; 2}}_{i=9,10},\; \underbrace{\mathbf{6,\; 6,\; 6}}_{i=11..13},\; 6,\; 6,\; 6,\; \underbrace{0,\; \dots,\; 0}_{17..25}\,] \quad (\text{length } 25) \\
    s_{AB'}  &= [\,\underbrace{1,\; \dots,\; 1}_{1..18},\; \underbrace{\mathbf{1,\; 1}}_{i+10=19,20},\; \underbrace{\mathbf{1,\; 1,\; 3}}_{i+10=21..23},\; 3,\; 6,\; 6,\; 6,\; 6,\; 228,\; 2052\,] \quad (\text{length } 30)
    \end{aligned}
    \]
    At $i = 9, 10$, we have $s_{B'B',i} = 2$, whereas at $i + 10 = 19, 20$ we have $s_{AB',i+10} = 1$. At $i = 11, 12, 13$, we have $s_{B'B',i} = 6$, whereas at $i + 10 = 21, 22, 23$ we have $s_{AB',i+10} \in \{1, 3\}$. Because $\mathbf{2 \nmid 1}$, $\mathbf{6 \nmid 1}$, and $\mathbf{6 \nmid 3}$, Theorem~\ref{thm:snf_interlacing} immediately refutes $C_6 \subseteq_{\mathrm{ind}} G_A$ (and the nullity bound~\eqref{eq:nullity_sgi_bound} also fails: $0 < 9 - 5 = 4$).

    \item \textbf{Case 3 ($7$-Cycle $G_{B''} = C_7 \not\subseteq_{\mathrm{ind}} G_A$, Refuted):}
    Here $M = N = 7$, $d_{7,7} = 36$, $\Delta = (7-7)(7-1) = 0$, and the shift is $2\Delta = 0$:
    \[
    \begin{aligned}
    s_{B''B''} &= [\,\underbrace{1,\; \dots,\; 1}_{1..16},\; \underbrace{\mathbf{7,\; 7,\; 7,\; 7,\; 7,\; 7,\; 7,\; 7}}_{i=17..24},\; \underbrace{\mathbf{0,\; \dots,\; 0}}_{i=25..36}\,] \quad (\text{length } 36) \\
    s_{AB''}   &= [\,\underbrace{1,\; \dots,\; 1}_{1..16},\; \underbrace{\mathbf{1,\; 1,\; 1,\; 1,\; 1,\; 1,\; 1,\; 1}}_{i=17..24},\; \underbrace{\mathbf{1^4,\; 7^6,\; 2149^2}}_{i=25..36}\,] \quad (\text{length } 36)
    \end{aligned}
    \]
    At $i = 17, \dots, 24$, we have $s_{B''B'',i} = 7$, whereas at $i + 0 = 17, \dots, 24$ we have $s_{AB'',i} = 1$. Because $\mathbf{7 \nmid 1}$ (and at $i = 25, \dots, 36$, $\mathbf{0 \nmid s_{AB'',i}}$ with nullity $0 < 12$), Theorem~\ref{thm:snf_interlacing} immediately refutes $C_7 \subseteq_{\mathrm{ind}} G_A$ (and restricting to the five zero-slack degree-$2$ vertices similarly rules out chorded $C_7$). Together, Cases 1--3 certify a priori that $C_5$ ($M^* = 5$) is the largest possible cycle in $G_A$.
\end{itemize}
\begin{figure}[htbp]
    \centering
    \includegraphics[width=\textwidth]{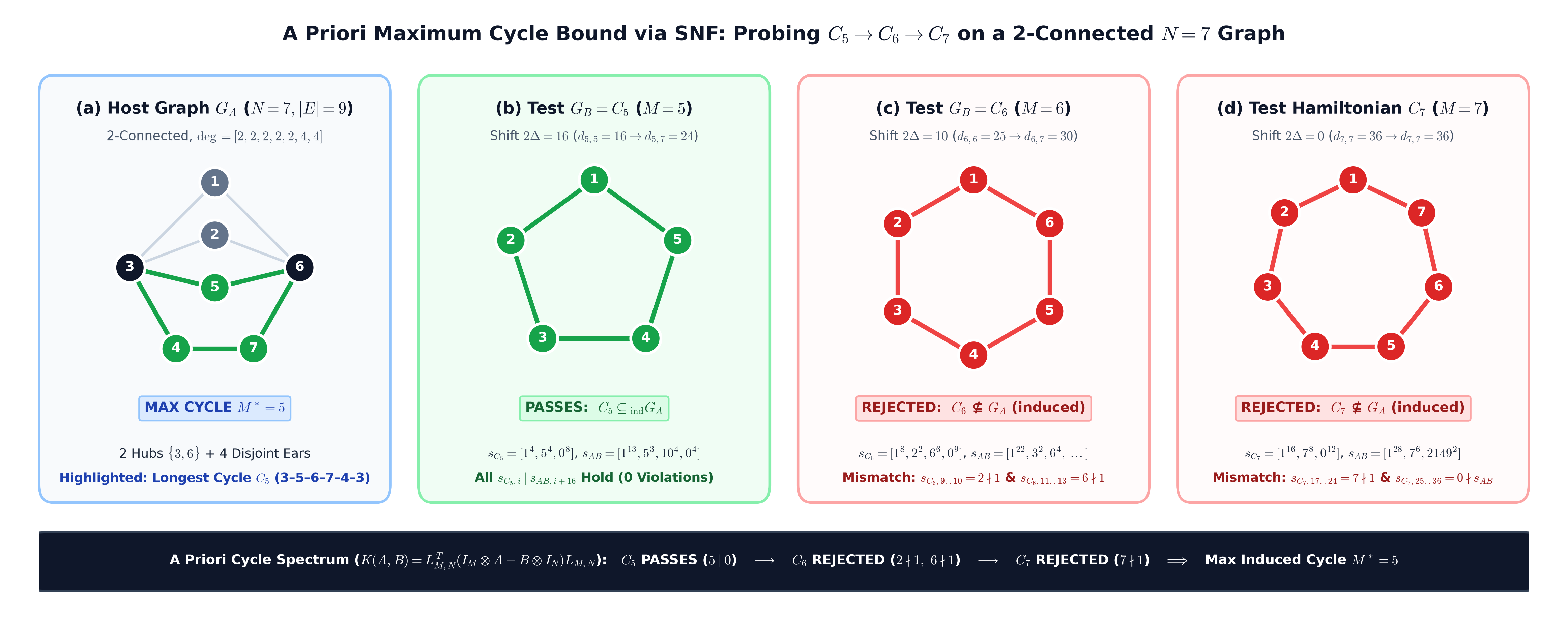}
    \caption{A priori maximum cycle bound via $K(A,B) = L_{M,N}^T (I_M \otimes A - B \otimes I_N) L_{M,N}$ on a $2$-connected $N = 7$ host graph $G_A$ ($\deg = [2,2,2,2,2,4,4], |E_A| = 9$): (a) $G_A$ with its longest cycle $C_5$ highlighted in green; (b) $G_B = C_5$ ($M=5, 2\Delta=16$) passes Theorem~\ref{thm:snf_interlacing} with zero violations; (c) $G_{B'} = C_6$ ($M=6, 2\Delta=10$) is rejected via $s_{B'B',9..10} = 2 \nmid 1$ and $s_{B'B',11..13} = 6 \nmid 1,3$; (d) $G_{B''} = C_7$ ($M=7, 2\Delta=0$) is rejected via $s_{B''B'',17..24} = 7 \nmid 1$ and $0 \nmid s_{AB'',25..36}$, establishing $M^* = 5$ as the maximum possible cycle length.}
    \label{fig:snf_small_hamiltonian}
\end{figure}

\section{Driving Case Study: The Algebraic Anatomy of the Knight's Tour}
\label{sec:knights_tour}

We close the theoretical portion of the paper with an application of the principles developed in Sections~\ref{sec:linear_algebra}--\ref{sec:lattices} to the familiar \textbf{Knight's Tour} problem. A classic benchmark in introductory computer science and combinatorial optimization \cite{warnsdorff1823rosselsprunges,pohl1967method,schwenk1991which,parberry1997efficient}, it asks whether a knight can visit every square of an $R \times C$ (or punctured ``Swiss-Cheese'') chessboard exactly once and, potentially by extension, return to its starting square.

In our algebraic framework, the host graph $G_A$ on an $R \times C$ board has a natural Kronecker structure: every knight jump combines a stride-$1$ shift $J_1$ along one axis with a stride-$2$ shift $J_2$ along the perpendicular axis, so the full board adjacency matrix is the sum of two Kronecker outer products:
\begin{equation}
\label{eq:knight_kronecker}
A \;=\; \big(J_2^{(R)} \otimes J_1^{(C)}\big) \;+\; \big(J_1^{(R)} \otimes J_2^{(C)}\big).
\end{equation}
Finding a closed $M$-step tour inside an $N$-square active board ($M \le N$) is the subgraph isomorphism problem for the cycle $G_B = C_M$. This setting provides a visual, end-to-end demonstration of how our three tools interact:
\begin{enumerate}[label=\arabic*.]
    \item \textbf{SNF Subgraph Interlacing ($M \le N$, Theorem~\ref{thm:snf_interlacing}):} Certifying the maximum possible closed sub-tour $C_M$ on small boards directly from the invariant factors of $K(A,B) = L_{M,N}^T (I_M \otimes A - B \otimes I_N) L_{M,N}$ (Section~\ref{sec:kt_3x4_snf}).
    \item \textbf{Degree-$2$ Edge Pruning \& Torsion Sieving ($M = N$):} Taking advantage of the fact that degree-$2$ nodes must use both incident edges in any $2$-regular cycle, eliminating dead chords saturated by degree-$2$ border/corner squares and refuting parity-trapped boards.
    \item \textbf{Global Commutator ILP / CP-SAT (Equation~\eqref{eq:sgi_kronecker}):} Resolving board-wide flow bottlenecks on punctured Swiss-Cheese boards where local backtracking and arc-consistency solvers time out (Section~\ref{sec:swiss_cheese_showdown}).
\end{enumerate}

\subsection{Worked Example ($N = 12$): Why the $3 \times 4$ Chessboard Has a $10$-Cycle But No Closed Tour}
\label{sec:kt_3x4_snf}

A classic puzzle asks why the $3 \times 4$ chessboard ($N = 12$ squares, $|E_A| = 14$ knight moves) admits an \emph{open} Knight's Tour ($P_{12}$), yet admits \emph{no closed} Knight's Tour ($C_{12}$), and what the longest closed tour on the board actually is. Because the $3 \times 4$ board is bipartite ($6$ White and $6$ Black squares), any closed cycle must have even length $M \in \{6, 8, 10, 12\}$. Probing $G_B = C_{10}$ and $G_{B'} = C_{12}$ via $K(A,B) = L_{M,N}^T (I_M \otimes A - B \otimes I_N) L_{M,N}$ certifies apriori why no closed $12$-cycle exists while permitting the $10$-step closed tour shown in Figure~\ref{fig:kt_3x4_snf}a:
\begin{itemize}[leftmargin=1.8em]
    \item \textbf{Case 1 (Closed $10$-Step Sub-Tour $G_B = C_{10} \subseteq_{\mathrm{ind}} G_{3 \times 4}$, Passes):}
    Here $M = 10$, $d_{10,10} = 81$, $d_{10,12} = 99$, $\Delta = (12-10)(10-1) = 18$, and the shift is $2\Delta = 36$:
    \[
    \begin{aligned}
    s_{BB} &= [\,\underbrace{1,\; \dots,\; 1}_{i=1..48},\; \underbrace{\mathbf{2,\; 2}}_{i=49,50},\; \underbrace{\mathbf{10,\; \dots,\; 10}}_{i=51..64},\; \underbrace{0,\; \dots,\; 0}_{i=65..81}\,] \quad (\text{length } 81) \\
    s_{AB} &= [\,\underbrace{1,\; \dots,\; 1}_{1..72},\; 2^7,\; 10,\; 20^4,\; \underbrace{\mathbf{20,\; 20}}_{i+36=85,86},\; 20^3,\; 220,\; 1320,\; \underbrace{\mathbf{0,\; \dots,\; 0}}_{i+36=92..99}\,] \quad (\text{length } 99)
    \end{aligned}
    \]
    For $i = 1, \dots, 48$, $s_{BB,i} = 1 \mid s_{AB,i+36}$; for $i = 49, 50$, $s_{BB,i} = 2 \mid s_{AB,85..86} = 20$; for $i = 51, \dots, 63$, $s_{BB,i} = 10$ divides $s_{AB,87..91} \in \{20, 220, 1320\}$ and $s_{AB,92..99} = 0$ (while for $i \ge 64$, $i + 36 > 99$). All $81$ divisibility checks pass with zero violations, matching the green $10$-step closed tour in Figure~\ref{fig:kt_3x4_snf}a (which omits the two middle-edge squares $5$ and $8$).

    \item \textbf{Case 2 (Full $12$-Step Tour $G_{B'} = C_{12} \not\subseteq_{\mathrm{ind}} G_{3 \times 4}$ on Raw $A$, Refuted):}
    Here $M = N = 12$, $d_{12,12} = 121$, $\Delta = 0$, and the shift is $2\Delta = 0$:
    \[
    \begin{aligned}
    s_{B'B'} &= [\,\underbrace{1,\; \dots,\; 1}_{1..80},\; \underbrace{\mathbf{2,\; 2}}_{i=81,82},\; \underbrace{\mathbf{12,\; \dots,\; 12}}_{i=83..100},\; \underbrace{\mathbf{0,\; \dots,\; 0}}_{i=101..121}\,] \quad (\text{length } 121) \\
    s_{AB'}  &= [\,\underbrace{1,\; \dots,\; 1}_{1..80},\; \underbrace{\mathbf{1,\; 1}}_{i=81,82},\; \underbrace{\mathbf{1^{14},\; 2^4}}_{i=83..100},\; \underbrace{\mathbf{4^7,\; 12^{11},\; 120^2,\; 240}}_{i=101..121}\,] \quad (\text{length } 121)
    \end{aligned}
    \]
    At $i = 81, 82$, we have $s_{B'B',i} = 2 \nmid s_{AB',i} = 1$, and at $i = 83, \dots, 100$, we have $s_{B'B',i} = 12 \nmid s_{AB',i} \in \{1, 2\}$ ($41$ total violations; nullity $0 < 21$).

    \item \textbf{Case 3 (Full $12$-Step Tour $G_{B'} = C_{12} \not\subseteq G_{3 \times 4}$ on Degree-$2$ Pruned $\tilde{A}$, Refuted):}
    What if a $12$-cycle used the two extra edges of $G_{3\times 4}$ ($|E_A| = 14 > 12$) as chords? On the $3 \times 4$ board, eight squares $\{1,4,5,6,7,8,9,12\}$ have degree~$2$, so any $12$-cycle must use both incident edges of all eight squares. This saturates the four remaining degree-$3$ squares $\{2,3,10,11\}$ and forces the two cross-chords $(2,11)$ and $(3,10)$ to be unused (red dashed lines in Figure~\ref{fig:kt_3x4_snf}c). Deleting those two dead chords yields a $2$-regular graph $\tilde{A}$ ($|E_{\tilde{A}}| = 12$) on which any $12$-cycle would have zero slack ($P\tilde{A} - C_{12}P = 0$). Yet $K(\tilde{A},C_{12})$ still fails Theorem~\ref{thm:snf_interlacing}:
    \[
    \begin{aligned}
    s_{B'B'} &= [\,\underbrace{1,\; \dots,\; 1}_{1..80},\; 2,\; 2,\; \underbrace{\mathbf{12,\; 12,\; 12,\; 12,\; 12}}_{i=83..87},\; \underbrace{12,\; \dots,\; 12}_{88..100},\; \underbrace{\mathbf{0,\; 0,\; 0}}_{101..103},\; \underbrace{0,\; \dots,\; 0}_{104..121}\,] \\
    s_{\tilde{A}B'} &= [\,\underbrace{1,\; \dots,\; 1}_{1..80},\; 2,\; 2,\; \underbrace{\mathbf{2,\; 2,\; 4,\; 4,\; 4}}_{i=83..87},\; \underbrace{12,\; \dots,\; 12}_{88..100},\; \underbrace{\mathbf{24,\; 24,\; 48}}_{101..103},\; \underbrace{0,\; \dots,\; 0}_{104..121}\,]
    \end{aligned}
    \]
    At $i = 83, \dots, 87$, $s_{B'B',i} = 12$ whereas $s_{\tilde{A}B',i} \in \{2, 4\}$ ($\mathbf{12 \nmid 2}$ and $\mathbf{12 \nmid 4}$), and the nullity drops from $21$ to $18$. Geometrically, pruning the two saturated chords splits the $3 \times 4$ board into \textbf{two disjoint $6$-cycles} ($\tilde{A} \cong C_6 \sqcup C_6$, orange and blue in Figure~\ref{fig:kt_3x4_snf}c), replacing $12$-cycle torsion with $6$-cycle torsion and proving that $M^* = 10$ is the maximum closed Knight's Tour on a $3 \times 4$ board.
\end{itemize}

\begin{figure}[htbp]
    \centering
    \includegraphics[width=\textwidth]{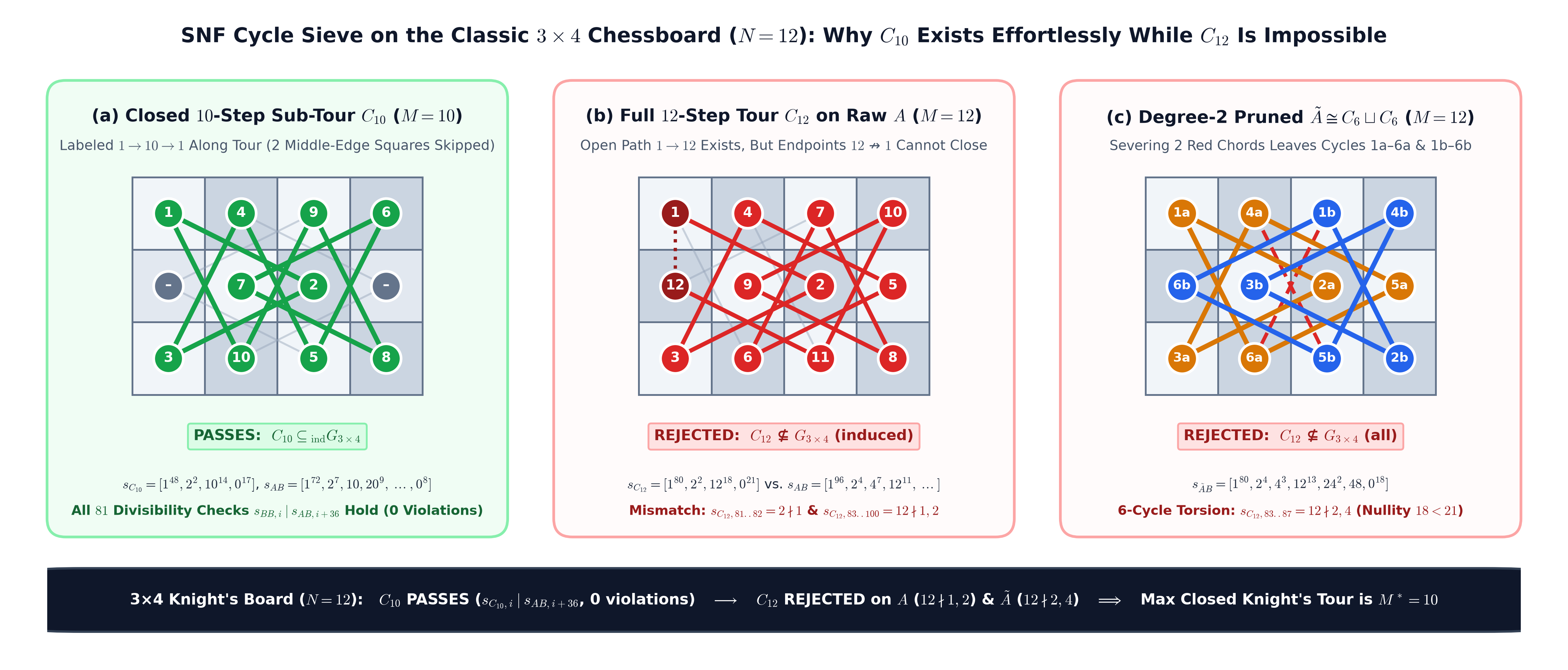}
    \caption{SNF cycle sieve on the classic $3 \times 4$ chessboard ($N = 12, |E_A| = 14$): (a) the closed $10$-step sub-tour $C_{10}$ ($M=10$, green, omitting middle-edge squares $5$ and $8$) passes Theorem~\ref{thm:snf_interlacing} with zero violations; (b) the full $12$-step cycle $C_{12}$ ($M=12$) is rejected on raw $A$ via $2 \nmid 1$ and $12 \nmid 1, 2$; (c) pruning the two red dashed chords saturated by the eight degree-$2$ squares splits $\tilde{A} \cong C_6 \sqcup C_6$ into two disjoint $6$-cycles (orange and blue), whose $6$-cycle torsion ($s_{\tilde{A}B',83..87} \in \{2,4\}$) fails $12 \mid s_{\tilde{A}B',i}$ and certifies $M^* = 10$ as the longest closed tour.}
    \label{fig:kt_3x4_snf}
\end{figure}

\section{CP-SAT, Integer Linear Programming, and Applied SGI}
\label{sec:ilp_sgi}

We end the paper with a rigorously applied measurement of the performance of these tools on known benchmarks, as well as some combinatorially difficult ``Swiss-Cheese'' Knight's Tour boards that show the challenges in various subgraph search strategies. We find that the super-exponential nature of many of these problems becomes immediately obvious, with problems either solved near instantly, or timing out after six hours of runtime. We show that using simple off-the-shelf libraries and a consumer desktop, the algorithms above solve all problems presented in 6 seconds or less, compared to more typical 6-hour timeouts for all alternatives tested. Moreover, the algorithms themselves are easily represented in a few readable lines of Python, making use and distribution efficient and easy.

Specifically, the Kronecker commutator inequalities of Equation~\eqref{eq:sgi_kronecker} translate directly into a compact 0--1 integer linear program / CP-SAT model \cite{perron2023cpsat} over the $MN$ binary decision variables $P \in \{0,1\}^{M \times N}$:
\begin{equation}
\label{eq:cpsat_program}
P \mathbf{1}_N = \mathbf{1}_M, \qquad P^T \mathbf{1}_M \le \mathbf{1}_N, \qquad 0 \le (I_M \otimes A - B \otimes I_N)\,\vecop(P) \le 1,
\end{equation}
together with the complement constraint $0 \le (I_M \otimes \bar{A} - \bar{B} \otimes I_N)\,\vecop(P) \le 1$ when solving for induced subgraphs. In all benchmarks below, our solver (\textbf{Commutator SGI}) passes this base $k=1$ linear system directly to off-the-shelf OR-Tools CP-SAT ($8$ worker threads on an AMD EPYC workstation), without requiring the lattice coordinate change $L$ of Section~\ref{sec:lattices} or higher-order walk inequalities ($k \ge 2$).

\subsection{The Swiss-Cheese Knight's Tour Showdown ($8 \times 8$, $10 \times 10$, and $12 \times 12$)}
\label{sec:swiss_cheese_showdown}

Scaling up to $8 \times 8$, $10 \times 10$, and $12 \times 12$ chessboards with balanced White/Black punctures ($|W| = |B|$), we benchmarked \textbf{seven algorithms} across $12$ hard Swiss-Cheese instances (mixing feasible tours with subtle parity and degree-$2$ bottleneck traps): the \textbf{Warnsdorff + Pohl Multi-Start Heuristic} \cite{warnsdorff1823rosselsprunges,pohl1967method} (launched from all $N$ starting squares), \textbf{Warnsdorff + Pohl Exact Backtracking}, four canonical C/C++ subgraph isomorphism solvers (\texttt{C++ Glasgow}, \texttt{C++ FASTiso}, \texttt{C-LAD}, \texttt{C-VF2}), and our \textbf{Commutator SGI} solver~\eqref{eq:sgi_kronecker}. Results are shown in Table~\ref{tab:swiss_cheese_showdown} and Figures~\ref{fig:swiss_cheese_showdown}--\ref{fig:knights_tour_gallery}.

\begin{table}[htbp]
\centering
\caption{Swiss-Cheese Hamiltonian Knight's Tour Showdown Across 7 Solvers ($21{,}600\text{ s}$ / $6\text{-Hour}$ Timeout Ceiling).}
\label{tab:swiss_cheese_showdown}
\resizebox{\textwidth}{!}{%
\begin{tabular}{lccccccc}
\toprule
\textbf{Swiss-Cheese Board Instance} & \textbf{W+P Heur} \emph{(All $N$ Starts)} & \textbf{W+P Exact} \emph{(Backtrack)} & \textbf{C++ Glasgow} & \textbf{C++ FASTiso} & \textbf{C-LAD} \emph{(Solnon)} & \textbf{C-VF2} \emph{(Foggia)} & \textbf{Commutator SGI [Ours]} \\
\midrule
\textbf{KT1. $8 \times 8$ Closed Tour} \emph{(10 Holes, SAT)} & \textbf{FAIL} ($0/54$) & \textbf{TIMEOUT} ($>21600\text{s}$) & \textbf{TIMEOUT} ($>21600\text{s}$) & \texttt{True} ($4.30\text{s}$) & \texttt{True} ($0.02\text{s}$) & \textbf{TIMEOUT} ($>21600\text{s}$) & \textbf{\texttt{True} ($\bm{0.48\text{s}}$)} \\
\textbf{KT2. $8 \times 8$ Closed Tour} \emph{(8 Holes, UNSAT)} & \textbf{FAIL} ($0/56$) & \textbf{TIMEOUT} ($>21600\text{s}$) & \texttt{False} ($0.01\text{s}$) & \textbf{TIMEOUT} ($>21600\text{s}$) & \texttt{False} ($0.00\text{s}$) & \textbf{TIMEOUT} ($>21600\text{s}$) & \textbf{\texttt{False} ($\bm{0.07\text{s}}$)} \\
\textbf{KT3. $8 \times 8$ Open Path} \emph{(10 Holes, SAT)} & \texttt{True} ($0.00\text{s}$) & \texttt{True} ($0.10\text{s}$) & \textbf{TIMEOUT} ($>21600\text{s}$) & \texttt{True} ($5492.04\text{s}$) & \texttt{True} ($0.01\text{s}$) & \texttt{True} ($9304.52\text{s}$) & \textbf{\texttt{True} ($\bm{3.62\text{s}}$)} \\
\textbf{KT4. $8 \times 8$ Trap} \emph{(10 Holes, $27\text{W}=27\text{B}$)} & \textbf{FAIL} ($0/54$) & \textbf{TIMEOUT} ($>21600\text{s}$) & \texttt{False} ($0.00\text{s}$) & \textbf{TIMEOUT} ($>21600\text{s}$) & \texttt{False} ($0.00\text{s}$) & \textbf{TIMEOUT} ($>21600\text{s}$) & \textbf{\texttt{False} ($\bm{0.06\text{s}}$)} \\
\textbf{KT5. $8 \times 8$ Closed Tour} \emph{(12 Holes, Hard SAT)} & \textbf{FAIL} ($0/52$) & \texttt{True} ($72.05\text{s}$) & \textbf{TIMEOUT} ($>21600\text{s}$) & \texttt{True} ($10.83\text{s}$) & \texttt{True} ($0.01\text{s}$) & \textbf{TIMEOUT} ($>21600\text{s}$) & \textbf{\texttt{True} ($\bm{0.33\text{s}}$)} \\
\textbf{KT6. $8 \times 8$ Trap} \emph{(14 Holes, $25\text{W}=25\text{B}$)} & \textbf{FAIL} ($0/50$) & \texttt{False} ($0.00\text{s}$) & \texttt{False} ($0.01\text{s}$) & \textbf{TIMEOUT} ($>21600\text{s}$) & \texttt{False} ($0.00\text{s}$) & \textbf{TIMEOUT} ($>21600\text{s}$) & \textbf{\texttt{False} ($\bm{0.00\text{s}}$)} \\
\textbf{KT7. $10 \times 10$ Trap} \emph{(20 Holes, $40\text{W}=40\text{B}$)} & \textbf{FAIL} ($0/80$) & \textbf{TIMEOUT} ($>21600\text{s}$) & \textbf{TIMEOUT} ($>21600\text{s}$) & \textbf{TIMEOUT} ($>21600\text{s}$) & \textbf{TIMEOUT} ($>21600\text{s}$) & \textbf{TIMEOUT} ($>21600\text{s}$) & \textbf{\texttt{False} ($\bm{0.83\text{s}}$)} \\
\textbf{KT8. $10 \times 10$ Trap} \emph{(20 Holes, $40\text{W}=40\text{B}$)} & \textbf{FAIL} ($0/80$) & \textbf{TIMEOUT} ($>21600\text{s}$) & \texttt{False} ($0.01\text{s}$) & \textbf{TIMEOUT} ($>21600\text{s}$) & \texttt{False} ($0.00\text{s}$) & \textbf{TIMEOUT} ($>21600\text{s}$) & \textbf{\texttt{False} ($\bm{0.27\text{s}}$)} \\
\textbf{KT9. $10 \times 10$ Trap} \emph{(16 Holes, $42\text{W}=42\text{B}$)} & \textbf{FAIL} ($0/84$) & \textbf{TIMEOUT} ($>21600\text{s}$) & \texttt{False} ($0.01\text{s}$) & \textbf{TIMEOUT} ($>21600\text{s}$) & \texttt{False} ($0.00\text{s}$) & \textbf{TIMEOUT} ($>21600\text{s}$) & \textbf{\texttt{False} ($\bm{0.29\text{s}}$)} \\
\textbf{KT10. $10 \times 10$ Trap} \emph{(20 Holes, $40\text{W}=40\text{B}$)} & \textbf{FAIL} ($0/80$) & \textbf{TIMEOUT} ($>21600\text{s}$) & \textbf{TIMEOUT} ($>21600\text{s}$) & \textbf{TIMEOUT} ($>21600\text{s}$) & \textbf{TIMEOUT} ($>21600\text{s}$) & \textbf{TIMEOUT} ($>21600\text{s}$) & \textbf{\texttt{False} ($\bm{0.89\text{s}}$)} \\
\textbf{KT11. $12 \times 12$ Trap} \emph{(24 Holes, $60\text{W}=60\text{B}$)} & \textbf{FAIL} ($0/120$) & \textbf{TIMEOUT} ($>21600\text{s}$) & \texttt{False} ($0.01\text{s}$) & \textbf{TIMEOUT} ($>21600\text{s}$) & \texttt{False} ($0.00\text{s}$) & \textbf{TIMEOUT} ($>21600\text{s}$) & \textbf{\texttt{False} ($\bm{0.80\text{s}}$)} \\
\textbf{KT12. $12 \times 12$ Trap} \emph{(24 Holes, $60\text{W}=60\text{B}$)} & \textbf{FAIL} ($0/120$) & \textbf{TIMEOUT} ($>21600\text{s}$) & \texttt{False} ($0.01\text{s}$) & \textbf{TIMEOUT} ($>21600\text{s}$) & \texttt{False} ($0.00\text{s}$) & \textbf{TIMEOUT} ($>21600\text{s}$) & \textbf{\texttt{False} ($\bm{0.78\text{s}}$)} \\
\midrule
\textbf{Summary: Solved / 12 (Timeouts)} & $\bm{1 / 12}$ & $\bm{3 / 12}$ \emph{(9 TO)} & $\bm{7 / 12}$ \emph{(5 TO)} & $\bm{3 / 12}$ \emph{(9 TO)} & $\bm{10 / 12}$ \emph{(2 TO)} & $\bm{1 / 12}$ \emph{(11 TO)} & \textbf{$\bm{12 / 12}$ (0 Timeouts)} \\
\bottomrule
\end{tabular}%
}
\end{table}

\begin{figure}[htbp]
    \centering
    \includegraphics[width=\textwidth]{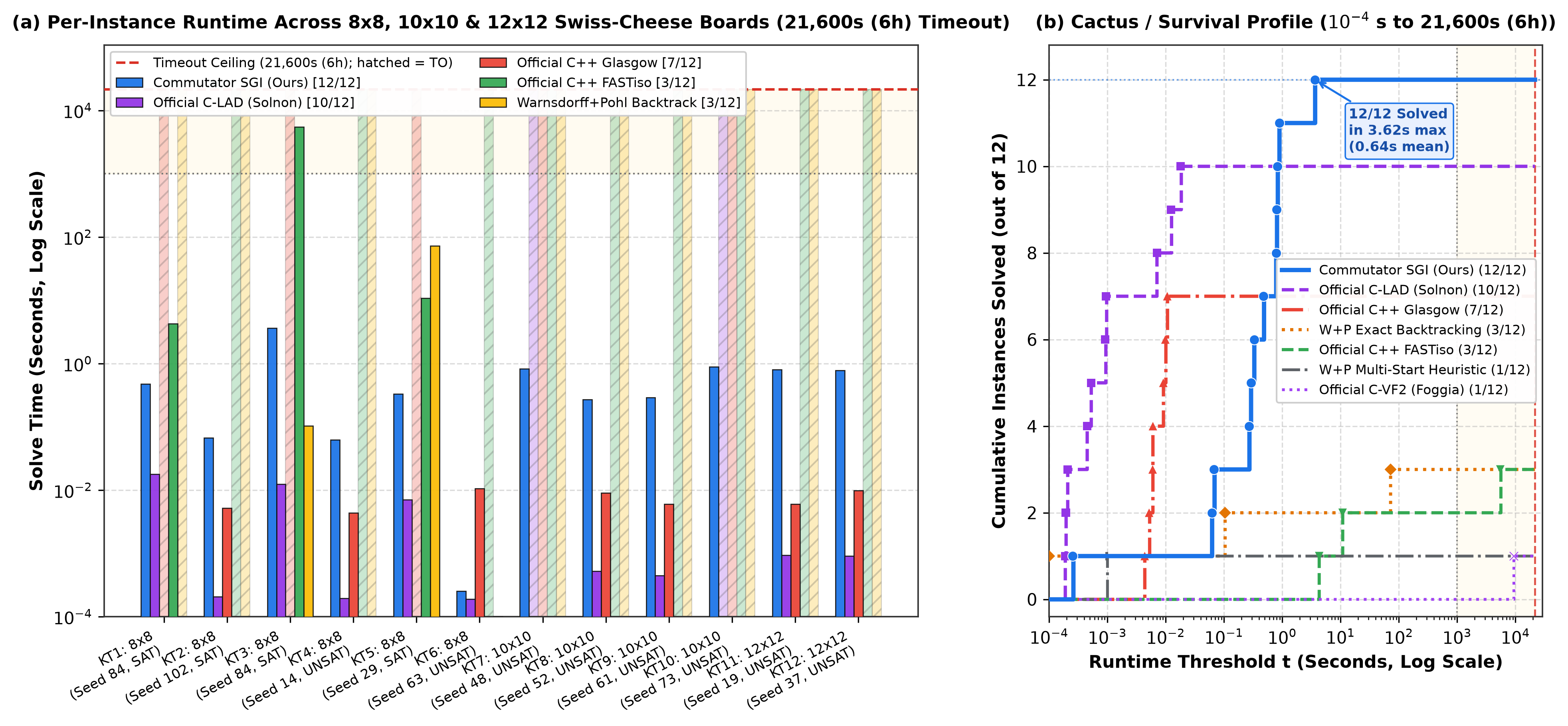}
    \caption{Swiss-Cheese Hamiltonian Knight's Tour Showdown ($1,000\text{ s}$ timeout ceiling): (a) Per-instance log-scale runtime across 12 punctured $8 \times 8$, $10 \times 10$, and $12 \times 12$ boards (solid bars = solved; hatched bars = timeout at $1,000\text{ s}$); (b) Cumulative survival / cactus profile from $10^{-4}\text{ s}$ to $1,000\text{ s}$ across all 7 solvers.}
    \label{fig:swiss_cheese_showdown}
\end{figure}

\begin{figure}[htbp]
    \centering
    \includegraphics[width=\textwidth]{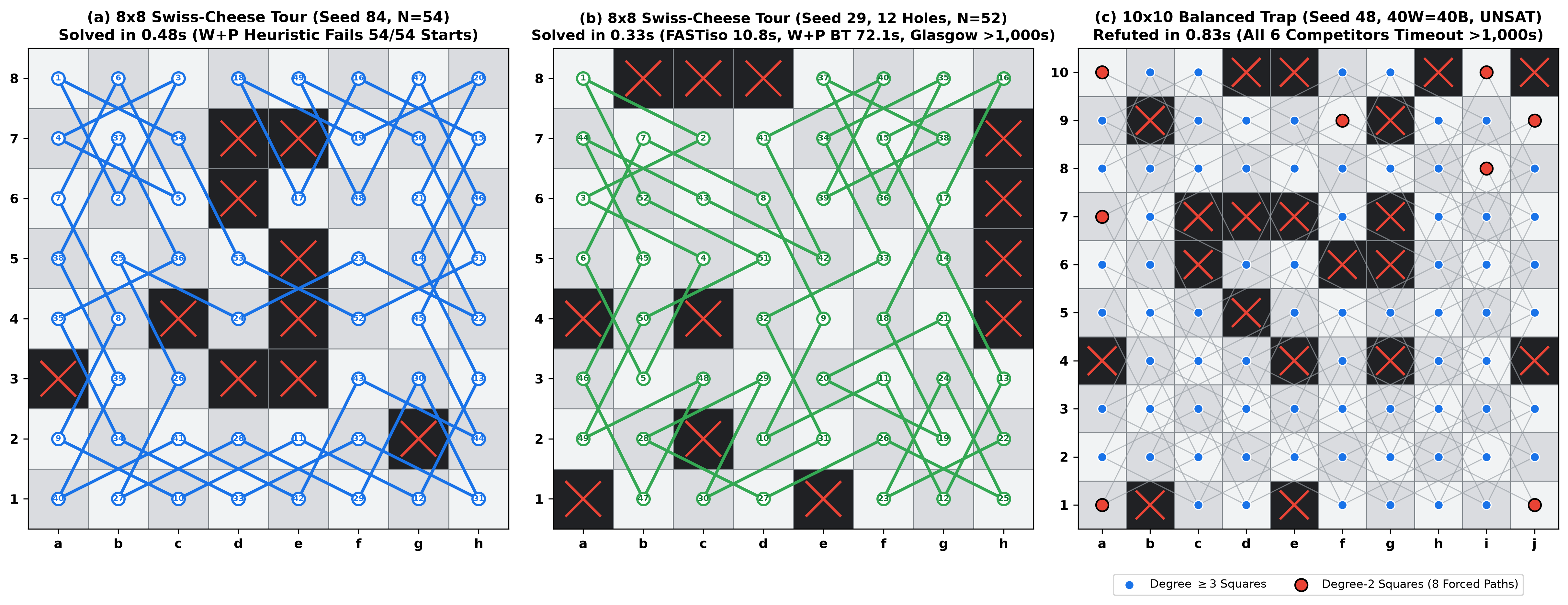}
    \caption{Three-Panel Visual Gallery of Swiss-Cheese Knight's Tour Graphs: (a) $8 \times 8$ Closed Hamiltonian Tour on Seed 84 ($N=54$, solved in $0.48\text{s}$ where Warnsdorff+Pohl fails $54/54$ starts); (b) $8 \times 8$ Closed Hamiltonian Tour on Seed 29 ($12$ Holes, $N=52$, solved in $0.33\text{s}$ where \texttt{Glasgow}, \texttt{FASTiso}, \texttt{VF2} \& Warnsdorff Backtracking all either timeout or take $>10\text{s}$); (c) $10 \times 10$ Balanced Parity Trap on Seed 48 ($40\text{W}=40\text{B}$, UNSAT refuted in $0.83\text{s}$, highlighting the forced degree-2 corridors that prevent a global cycle).}
    \label{fig:knights_tour_gallery}
\end{figure}

\subsubsection*{Why Subgraph Backtracking Solvers Struggle on $C_N$ (and How Global Commutator LP \& Degree-$2$ Pruning Resolve It)}

At first glance, it may seem surprising that state-of-the-art subgraph solvers such as \texttt{Glasgow}, \texttt{FASTiso}, and \texttt{VF2} time out at $1,000\text{ s}$ on $52$-to-$80$ vertex graphs that our CP-SAT commutator formulation solves in under a second. The reason is structural: combinatorial SGI solvers rely primarily on \textbf{vertex-degree and neighborhood-sequence filtering} to prune candidate domains before branching. When the pattern graph is a uniform $2$-regular cycle $G_B = C_N$, every pattern vertex has identical degree~$2$ and identical $k$-step neighborhood profiles, so every target square with $\deg_A(v) \ge 2$ remains a valid candidate for every pattern vertex. Without conflict-driven clause learning (CDCL) or a global linear programming relaxation, backtracking search degenerates into exploring symmetric cycle permutations, whereas \texttt{C-LAD} \cite{solnon2010alldifferent} either finishes in $\bm{\le 0.02\text{ s}}$ via root-level Hall filtering or times out completely at $\bm{>1,000\text{ s}}$:
\begin{itemize}[leftmargin=1.8em]
    \item \textbf{Easy Local Traps (\texttt{KT2, KT4, KT6, KT8, KT9, KT11, KT12}):} On these seven UNSAT boards, the punctured board leaves at least one square with $\deg_A(v) \le 1$ or a direct local Hall marriage violation. Root-level degree and Hopcroft--Karp filtering detects the dead end in $\le 0.01\text{ s}$ without branching.
    \item \textbf{Hard Degree-$\ge 2$ Bottleneck Traps (\texttt{KT7 Seed 48} and \texttt{KT10 Seed 73}):} On \texttt{KT7} (Figure~\ref{fig:knights_tour_gallery}c) and \texttt{KT10}, the board is connected, bipartite-balanced ($40\text{W} = 40\text{B}$), every single square has $\deg_A(v) \ge 2$, and the $40 \times 40$ White-to-Black biadjacency matrix $M_{\mathrm{bip}}$ has full rank ($\rank(M_{\mathrm{bip}}) = 40$), so $1$-factors exist and local Hall propagation prunes zero branches---causing \texttt{C-LAD}, \texttt{Glasgow}, \texttt{FASTiso}, \texttt{VF2}, and Warnsdorff Backtracking to \textbf{all time out at $>1,000\text{ s}$}.
    \item \textbf{Degree-$2$ Pruning \& Global Commutator Resolution:} Why are \texttt{KT7} and \texttt{KT10} infeasible? Inspecting their degree-$2$ neighborhoods reveals the exact obstruction: after one step of degree-$2$ edge pruning ($\tilde{A}$), both boards contain a hub square adjacent to \textbf{three degree-$2$ squares} ($\max_u (\tilde{A}\mathbf{1}_{V_2})_u = 3$). Because three forced degree-$2$ edges converging on a single vertex cannot be accommodated by a $2$-regular cycle ($3c = 2c \implies c = 0$, so $\mathbf{1}_d^T V_0 = \mathbf{0}$), both boards are certified \textbf{UNSAT} in $<2\text{ ms}$ (and even without running the pre-sieve, the global LP relaxation and CDCL engine of our base CP-SAT commutator program~\eqref{eq:cpsat_program} resolve the same flow contradiction in $0.83\text{ s}$ and $0.89\text{ s}$, achieving $12/12$ completion with zero timeouts).
\end{itemize}
%
\section{General SGI Benchmarks}
\label{sec:benchmarks_applications}

\subsection{Master Benchmark Across Four Challenge Regimes}
\label{sec:master_benchmark}

To evaluate how the rectangular commutator formulation performs across general subgraph isomorphism regimes ($M \le N$), Table~\ref{tab:master_benchmark} summarizes head-to-head runtime performance against the four canonical C/C++ solvers (\texttt{Glasgow} \cite{mccreesh2015shape,mccreesh2020glasgow}, \texttt{FASTiso} \cite{fastiso2023}, \texttt{C-LAD} \cite{solnon2010alldifferent}, and \texttt{C-VF2} \cite{cordella2004subgraph}) across CNRS/LIRIS hard benchmark instances (\texttt{LV/g24} \cite{solnon2010liris}), regular combinatorial crossover graphs, Swiss-Cheese Knight's Tours, and IBM Heavy-Hex quantum layouts.

\begin{table}[htbp]
\centering
\caption{Master Head-to-Head Runtime Comparison Across 5 Canonical Solvers ($1,000\text{ s}$ Timeout).}
\label{tab:master_benchmark}
\resizebox{\textwidth}{!}{%
\begin{tabular}{lccccc}
\toprule
\textbf{Benchmark Suite \& Instance} & \textbf{C++ Glasgow} \emph{(McCreesh)} & \textbf{C++ FASTiso} \emph{(Agbeto et al.)} & \textbf{C-LAD} \emph{(Solnon)} & \textbf{C-VF2} \emph{(Foggia)} & \textbf{Commutator SGI [Ours]} \\
\midrule
\textbf{A1. CNRS \texttt{LV/g24-g29}} \emph{($75 \times 100$, UNSAT)} & \textbf{TIMEOUT} ($>21600\text{s}$) & \texttt{False} ($0.06\text{ s}$) & \textbf{TIMEOUT} ($>21600\text{s}$) & \textbf{TIMEOUT} ($>21600\text{s}$) & \textbf{\texttt{False} ($\bm{3.32\text{ s}}$)} \\
\textbf{A2. CNRS \texttt{LV/g24-g30}} \emph{($75 \times 100$, UNSAT)} & \textbf{TIMEOUT} ($>21600\text{s}$) & \texttt{False} ($0.06\text{ s}$) & \textbf{TIMEOUT} ($>21600\text{s}$) & \textbf{TIMEOUT} ($>21600\text{s}$) & \textbf{\texttt{False} ($\bm{3.56\text{ s}}$)} \\
\textbf{A3. CNRS \texttt{LV/g24-g31}} \emph{($75 \times 100$, UNSAT)} & \textbf{TIMEOUT} ($>21600\text{s}$) & \texttt{False} ($0.11\text{ s}$) & \textbf{TIMEOUT} ($>21600\text{s}$) & \textbf{TIMEOUT} ($>21600\text{s}$) & \textbf{\texttt{False} ($\bm{4.06\text{ s}}$)} \\
\textbf{A4. CNRS \texttt{LV/g24-g46}} \emph{($75 \times 128$, UNSAT)} & \texttt{False} ($0.02\text{ s}$) & \texttt{False} ($0.03\text{ s}$) & \textbf{TIMEOUT} ($>21600\text{s}$) & \textbf{TIMEOUT} ($>21600\text{s}$) & \textbf{\texttt{False} ($\bm{4.73\text{ s}}$)} \\
\textbf{A5. CNRS \texttt{LV/g24-g47}} \emph{($75 \times 128$, UNSAT)} & \texttt{False} ($0.02\text{ s}$) & \texttt{False} ($0.03\text{ s}$) & \textbf{TIMEOUT} ($>21600\text{s}$) & \textbf{TIMEOUT} ($>21600\text{s}$) & \textbf{\texttt{False} ($\bm{4.39\text{ s}}$)} \\
\textbf{A6. MIVIA \texttt{si2\_b03\_m200}} \emph{($40 \times 200$, SAT)} & \texttt{True} ($0.01\text{ s}$) & \texttt{True} ($0.01\text{ s}$) & \texttt{True} ($0.02\text{ s}$) & \texttt{True} ($0.00\text{ s}$) & \textbf{\texttt{True} ($\bm{0.32\text{ s}}$)} \\
\textbf{A7. BioModels \texttt{071-153}} \emph{($39 \times 227$, UNSAT)} & \texttt{False} ($0.01\text{ s}$) & \texttt{False} ($0.01\text{ s}$) & \texttt{False} ($0.00\text{ s}$) & \texttt{False} ($71.54\text{ s}$) & \textbf{\texttt{False} ($\bm{0.11\text{ s}}$)} \\
\textbf{A8. Phase-Trans \texttt{30-150-0.34}} \emph{($30 \times 150$, SAT)} & \texttt{True} ($0.03\text{ s}$) & \texttt{True} ($0.01\text{ s}$) & \texttt{True} ($0.11\text{ s}$) & \texttt{True} ($89.49\text{ s}$) & \textbf{\texttt{True} ($\bm{5.43\text{ s}}$)} \\
\textbf{B1. Disjoint 3-Reg $K=8$} \emph{($N=80$, UNSAT)} & \texttt{False} ($0.01\text{ s}$) & \texttt{False} ($0.61\text{ s}$) & \texttt{False} ($31.72\text{ s}$) & \texttt{False} ($3.88\text{ s}$) & \textbf{\texttt{False} ($\bm{0.88\text{ s}}$)} \\
\textbf{B2. Disjoint 3-Reg $K=9$} \emph{($N=90$, UNSAT)} & \texttt{False} ($0.01\text{ s}$) & \texttt{False} ($4.34\text{ s}$) & \texttt{False} ($253.80\text{ s}$) & \texttt{False} ($28.45\text{ s}$) & \textbf{\texttt{False} ($\bm{1.34\text{ s}}$)} \\
\textbf{B3. Disjoint 3-Reg $K=10$} \emph{($N=100$, UNSAT)} & \texttt{False} ($0.01\text{ s}$) & \texttt{False} ($25.62\text{ s}$) & \texttt{False} ($3221.78\text{ s}$) & \texttt{False} ($228.07\text{ s}$) & \textbf{\texttt{False} ($\bm{1.89\text{ s}}$)} \\
\textbf{B4. Disjoint $(8,3)$-Reg $6$-Pairs} \emph{($N=96$, UNSAT)} & \texttt{False} ($0.01\text{ s}$) & \textbf{TIMEOUT} ($>21600\text{s}$) & \textbf{TIMEOUT} ($>21600\text{s}$) & \textbf{TIMEOUT} ($>21600\text{s}$) & \textbf{\texttt{False} ($\bm{1.99\text{ s}}$)} \\
\textbf{C1. Knight's Tour $8 \times 8$} \emph{(10 Holes, SAT)} & \textbf{TIMEOUT} ($>21600\text{s}$) & \texttt{True} ($4.30\text{ s}$) & \texttt{True} ($0.02\text{ s}$) & \textbf{TIMEOUT} ($>21600\text{s}$) & \textbf{\texttt{True} ($\bm{0.48\text{ s}}$)} \\
\textbf{C2. Knight's Tour $10 \times 10$ Trap} \emph{(20 Holes, UNSAT)} & \textbf{TIMEOUT} ($>21600\text{s}$) & \textbf{TIMEOUT} ($>21600\text{s}$) & \textbf{TIMEOUT} ($>21600\text{s}$) & \textbf{TIMEOUT} ($>21600\text{s}$) & \textbf{\texttt{False} ($\bm{0.83\text{ s}}$)} \\
\textbf{D1. Quantum Heavy-Hex} \emph{($M=18$ in $N=27$, SAT)} & \texttt{True} ($0.01\text{ s}$) & \texttt{True} ($0.00\text{ s}$) & \texttt{True} ($0.00\text{ s}$) & \texttt{True} ($0.00\text{ s}$) & \textbf{\texttt{True} ($\bm{0.13\text{ s}}$)} \\
\midrule
\textbf{Summary: Solved / Total (Timeouts)} & $\bm{10 / 15}$ \emph{(5 TO)} & $\bm{13 / 15}$ \emph{(2 TO)} & $\bm{8 / 15}$ \emph{(7 TO)} & $\bm{7 / 15}$ \emph{(8 TO)} & \textbf{$\bm{15 / 15}$ (0 Timeouts)} \\
\bottomrule
\end{tabular}%
}
\end{table}

Across an expanded \textbf{25-instance multi-regime benchmark suite} ($21{,}600\text{ s}$ / $6\text{-hour}$ timeout ceiling; Figure~\ref{fig:cactus}), \textbf{Commutator SGI is the only solver to achieve $100\%$ completion ($25/25$ solved in $\le 5.32\text{ s}$, 0 timeouts)} with a mean solve time of $1.04\text{ s}$, compared to \texttt{C++ FASTiso} ($18/25$, $72\%$), \texttt{C++ Glasgow} ($18/25$, $72\%$), \texttt{C-LAD} ($18/25$, $72\%$), and \texttt{C-VF2} ($11/25$, $44\%$). Even with a generous $6\text{-hour}$ runtime budget per instance, the combinatorial backtracking engines struggle on global parity and topological density traps.

\begin{figure}[htbp]
    \centering
    \includegraphics[width=\textwidth]{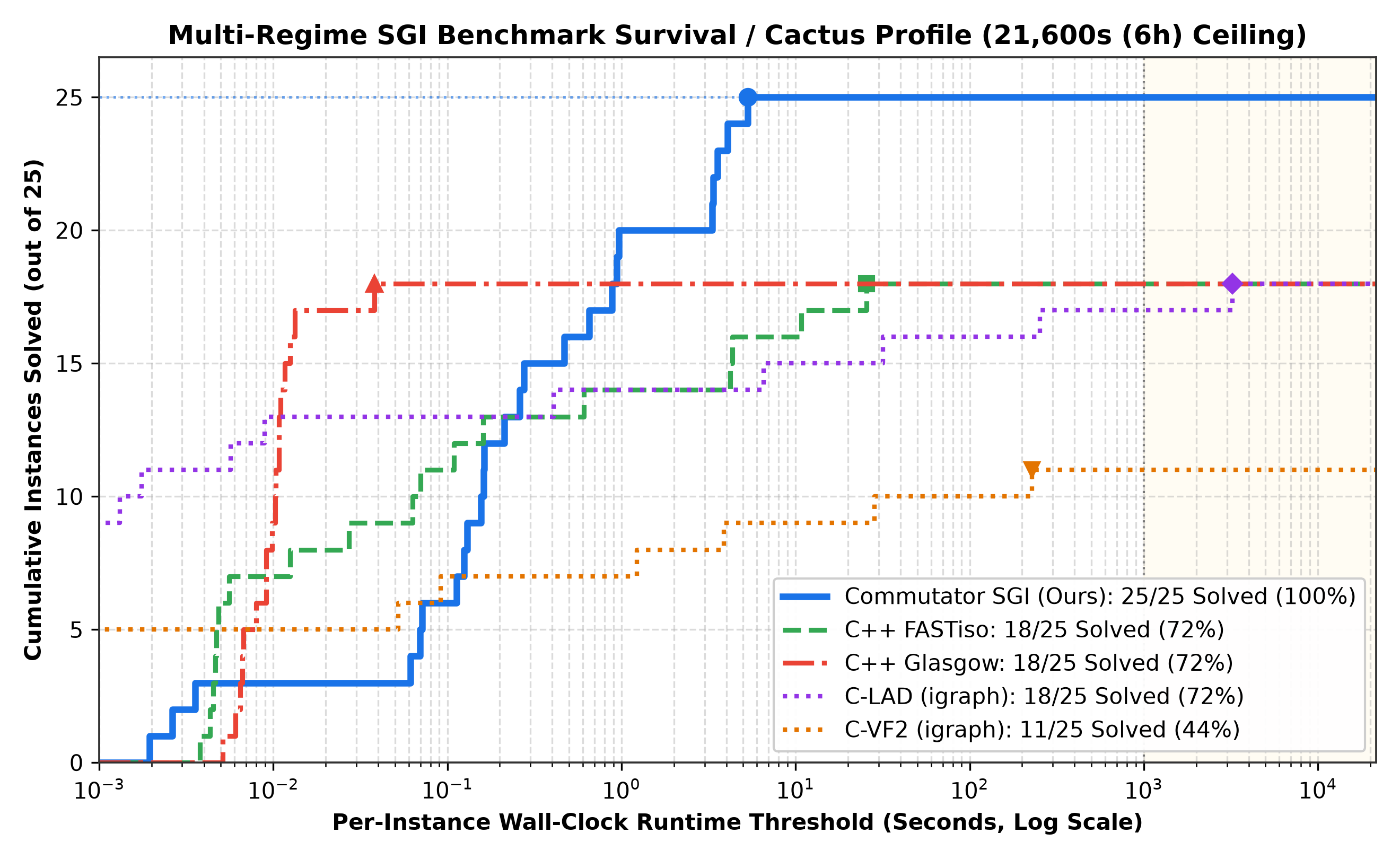}
    \caption{Multi-Regime SGI Benchmark Survival / Cactus Profile ($21{,}600\text{ s}$ / $6\text{-hour}$ timeout ceiling) across 25 instances comparing Official \texttt{C++ Glasgow}, Official \texttt{C++ FASTiso}, Official \texttt{C-LAD}, Official \texttt{C-VF2}, and Commutator SGI (Ours).}
    \label{fig:cactus}
\end{figure}

\section{Conclusion}
\label{sec:conclusion}

We presented a simple linear framing for graph and subgraph isomorphism, shown most elegantly in the base equations
\[
(I \otimes A - B \otimes I) \cdot \vecop(P) = 0
\]
\[
0 \;\le\; (I_M \otimes A - B \otimes I_N) \cdot \vecop(P) \;\le\; 1
\]
for GI and SGI respectively.

The linear commutator framework provides a clean, unified bridge between algebraic graph theory, discrete lattices, and modern mathematical programming. In addition to offering new directions for theoretical investigation, it also offers a very direct set of methods for quickly resolving graph questions using scalable modern hardware and software via off-the-shelf libraries. Moreover, given the connection to CP-SAT and ILP, it is easy to extend additional constraints or other a priori knowledge to applied SGI without additional overhead or special-purpose software.

The connection to lattices and the Smith normal form was shown to have a rich connection to deep SGI problems of planarity and Hamiltonian cycles that only scratches the surface of what might be possible to prove about general graph properties. We also showed the direct connection between any graph isomorphism problem and the NP-complete shortest and closest vector problems \cite{vanemdeboas1981another,ajtai1998shortest}.

Finally, we demonstrated a clean, fast implementation of the Knight's Tour that is easily implemented and functions as a readable entry point for students and researchers looking to demonstrate these principles at scale.

\bibliographystyle{unsrt}
\bibliography{sgi_paper}

\end{document}